\documentclass{article}

\usepackage[margin=1.0in]{geometry}
\usepackage[noblocks]{authblk}
\usepackage[english]{babel}
\usepackage{braket}
\usepackage{cancel, amsthm, amsfonts, amssymb, graphicx, booktabs, braket, bbold }
\usepackage[colorlinks=true,linkcolor=blue, citecolor= blue]{hyperref}
\usepackage{physics, mathtools, bbold, enumerate, caption, blkarray, breqn, dsfont}
\usepackage[boxed,ruled,vlined]{algorithm2e}
\usepackage[section]{placeins}
\usepackage[all]{xy}
\usepackage{transparent}
\usepackage{enumerate}
\usepackage[normalem]{ulem}
\usepackage[labelformat=simple]{subcaption}



\newtheorem{definition}{Definition}
\newtheorem{lemma}{Lemma}

\newtheorem{theorem}{Theorem}
\newtheorem*{theorem*}{Theorem}

\newtheorem{corollary}{Corollary}

\newtheorem*{mainres*}{Main Result}

\allowdisplaybreaks  

\newcommand{\z}{\mathbb{Z}}
\newcommand{\h}{\mathcal{H}}
\newcommand{\y}{\widehat{\mathcal{T}}}

\DeclareMathOperator{\Ch}{Ch}
\date{}

\newcommand{\p}{{\Pi}}

\usepackage[dvipsnames]{xcolor}
\usepackage{tikz}
\usepackage{tikz-cd}
\usetikzlibrary{arrows.meta,
                positioning,
                quotes, arrows
                }
\usetikzlibrary{spy,shapes,shapes.misc,decorations.pathmorphing,fit,backgrounds}
\usepackage[detect-all]{siunitx}
\usepackage{tkz-euclide}
\usepackage{relsize}
\tikzset{fontscale/.style = {font=\relsize{#1}}
    }

\begin{document}

\title{Compactly Supported Wannier Functions and Strictly Local Projectors}

\author[1]{Pratik Sathe\thanks{psathe@physics.ucla.edu}}
\author[1]{Fenner Harper}
\author[1]{Rahul Roy\thanks{rroy@physics.ucla.edu}}

\affil[1]{Mani L. Bhaumik Institute for Theoretical Physics, Department of Physics and Astronomy, \par University of California at Los Angeles, Los Angeles, California 90095, USA}

\maketitle

\begin{abstract}

Wannier functions that are maximally localized help in understanding many properties of crystalline materials. In the absence of topological obstructions, they are at least exponentially localized. In some cases such as flat-band Hamiltonians, it is possible to construct Wannier functions that are even more localized, so that they are compactly supported thus having zero support outside their corresponding locations. Under what general conditions is it possible to construct compactly supported Wannier functions? We answer this question in this paper. Specifically, we show that in 1d non-interacting tight-binding models, strict locality of the projection operator is a necessary and sufficient condition for a subspace to be spanned by a compactly supported orthogonal basis, independent of lattice translation symmetry. For any strictly local projector, we provide a procedure for obtaining such a basis. For higher dimensional systems, we discuss some additional conditions under which an occupied subspace is spanned by a compactly supported orthogonal basis, and show that the corresponding projectors are topologically trivial in many cases. We also show that a projector in arbitrary dimensions is strictly local if and only if for any chosen axis, its image is spanned by hybrid Wannier functions that are compactly supported along that axis.

\end{abstract}

\section{Introduction}
Extended Bloch wavefunctions and localized Wannier functions \cite{Wannier1937} are two common choices of basis vectors for a Bloch band. Localized Wannier functions have  applications in a number of fields, including the modern theory of polarization~\cite{Resta2010}, orbital magnetization~\cite{Thonhauser2005}, quantum transport~\cite{Calzolari2004} and tight binding interpolation~\cite{ashcroft1976solid, Yates2007}. Consequently, conditions required for the existence of localized Wannier functions have been investigated extensively. Isolated bands of 1d inversion symmetric systems are always spanned by exponentially localized Wannier functions as shown by Kohn~\cite{Kohn1959}. The localization properties of Wannier functions can often be inferred from the localization properties of the associated band projector. For instance, band projectors often possess real space matrix elements which decay exponentially~\cite{des1964energy}, leading to generalizations~\cite{dexCloizeauxexpWan, nenciu1983existence} of Kohn's result.  

Because of the importance of obtaining localized Wannier functions, there has been significant interest in obtaining Wannier functions that are as localized as possible. A popular variational approach seeks localized Wannier functions by numerically minimizing the second moment of the Wannier functions around their centers~\cite{marzari1997maximally}. It has been shown that maximally localized Wannier functions decay exponentially (or faster) in 1d systems, with extensions proved for higher dimensional systems in the absence of topological obstructions~\cite{Panati2013}.

A related, and sometimes more extreme form of wavefunction localization is compact support or strict localization in lattice models, wherein a wavefunction has non-zero support only over a finite set of orbitals of the lattice. Wavefunctions that are linear combinations of a finite set of orbitals bear a close analogy to Boys orbitals~\cite{Boys1960} which are studied in chemistry in the context of chemical bonding and other applications. The wavefunctions of Bloch electrons deep under the Fermi level are also expected to correspond to compactly supported (CS) Wannier functions. Additionally, in some applications, it is useful to approximate highly localized Wannier functions by CS wavefunctions~\cite{Koster1953, ozolicnvs2013compressed, budich2014search, barekat2014compressed}.

Non-orthogonal bases consisting of CS wavefunctions, also known as CS Wannier-type functions have also received significant attention~\cite{zheng2014exotic, dubail2015tensor, Read2017}. CS Wannier-type functions exist most commonly in strictly local (SL) flat-band Hamiltonians~\cite{chen2014impossibility} in the context of which they are also known as compact localized states (CLSs). Such bases help in understanding a number of many-body	quantum phenomena (see \cite{leykam2018artificial} for a review), including novel superconducting phases in multi-layer twisted graphene~\cite{Cao2018, Yankowitz2019}. CLSs have been used to classify and construct flat-band Hamiltonians~\cite{dias2015origami, maimaiti2017compact, maimaiti2019universal}. Models so constructed are often made interacting, in order to study interesting many-body quantum phenomena arising in such contexts. For example, orthogonal CLSs have been used to construct models with many-body localized states~\cite{Huber2010} and quantum scar states~\cite{Kuno2020b} in flat-band systems. Orthogonal CLSs that span an entire flat band are precisely CS Wannier functions of the flat band. Indeed, the conditions associated with the existence of orthogonal CLSs spanning flat bands and of CS Wannier functions in systems without flat bands are closely related.

Yet another variant of Wannier functions are hybrid Wannier functions~\cite{Sgiarovello2001}, that are localized and Wannier-like along one direction, and Bloch wave-like along the other directions. Localized hybrid Wannier functions have a number of applications, including the study of twisted bilayer graphene~\cite{Hejazi2021}, and characterization of static~\cite{Soluyanov2011, Taherinejad2014, Gresch2017} and Floquet topological insulators~\cite{Nakagawa2020}. Similar to CS Wannier functions, a set of hybrid Wannier functions that are CS along one of the axes can span a band in some cases. We refer to such functions as CS hybrid Wannier functions.

While conditions associated with the existence of exponentially localized Wannier functions, and of CLSs have been studied, those required for the existence of CS Wannier functions and CS hybrid Wannier functions remain unexplored. These considerations motivate us to pose the following questions: In a tight-binding lattice model, given an arbitrary set of occupied states, is there a way of determining whether their span possesses a CS Wannier basis? Analogously, in the absence of lattice translational invariance, under what conditions can the occupied subspace be spanned by an orthogonal basis of CS wavefunctions? We call such a basis a compactly supported orthogonal basis or a CSOB in short. We show that the localization properties of the associated projector has a direct bearing on these questions. For 1d systems, we answer this question completely, showing an equivalence between strict locality of an orthogonal basis and strict locality of the associated orthogonal projector. For higher dimensional systems, we obtain necessary and sufficient conditions for the existence of such a basis, as well as for CS hybrid Wannier functions. Our main results are summarized below.

\begin{mainres*} For an arbitrary subspace spanned by single particle states in a non-interacting tight-binding model, independent of translational invariance, the following statements are true.
\begin{enumerate}[(1)]
	\item In 1d systems, an orthogonal basis consisting of compactly supported wavefunctions (i.e. a CSOB) spanning the subspace exists iff. the associated orthogonal projector is strictly local.   
	\item For a lattice in $d$ dimensions, compactly supported hybrid Wannier functions localized along an axis exist for any  choice (out of $d$ possible choices) of the localization axis if and only if the associated band projector is strictly local.
	\item For arbitrary dimensional lattices, if the space is spanned by a CSOB, then the projector onto it is strictly local. If a projector is of a nearest neighbor form (or reducible to this form via a change of primitive vectors, or unit cell enlargement), its span possesses a CSOB.
\end{enumerate}	 \label{mainres}
\end{mainres*}

\begin{figure}[t]
        \centering
        \begin{tikzpicture}[auto, > = Stealth, 
           node distance = 10mm and 15mm,
              box/.style = {draw=black, minimum height=13mm, text width=30mm, align=center},
              bbox/.style = {fill=red!20, draw=black, minimum height=13mm, text width=30mm, align=center, double, rounded corners, draw=black},
              implies/.style={double, double equal sign distance, -implies},el/.style = {inner sep=5pt, align=left,font=\footnotesize},
	 every edge quotes/.style = {font=\footnotesize, align=center, inner sep=5pt}]

	\node (CWFone) [bbox] {CS Wannier basis/ CSOB};  
	
	\node (SLone) [box, left= of CWFone] {SL Projector};     
     \draw 	(CWFone) edge[implies-implies,double equal sign distance] (SLone);
     \node (oneDee) [left=of SLone] {\LARGE 1d:};
     
    \node (nDee) [below=2.75 of oneDee] {\LARGE $d>1$:};
    \node (CWFnd) [bbox, below= of CWFone] {CS Wannier basis/ CSOB};
    \node (hybridWan) [bbox, below =of  CWFnd] {CS hybrid Wannier basis for all axes};
    \node (ndNNP) [box, left = of CWFnd] {NN reducible projector};
	\node (SLnd) [box, below = of ndNNP] {SL projector};
	\draw 	(ndNNP) edge[implies] (CWFnd);
	\draw 	(ndNNP) edge[implies] (SLnd);
	\draw 	(CWFnd) edge[implies] (SLnd);
	\draw 	(SLnd) edge[implies-implies,double equal sign distance] (hybridWan);
	
	\node (torTriv) [box, below right= 0mm and 25mm of CWFnd] {Topological triviality};
	\path 	(CWFnd) edge[implies] node[el,above,sloped] {Even dimensions} (torTriv);
\path 	(hybridWan) edge[implies] node[el,below,sloped] {w/ T.I.} (torTriv);
\end{tikzpicture}
        \caption{Summary of the main results of this paper. All the statements hold true independent of whether the the system is translationally invariant (TI), except for the statement connecting the existence of CS hybrid Wannier basis for all axes to topological triviality.}
        \label{fig:results_summary}
\end{figure}
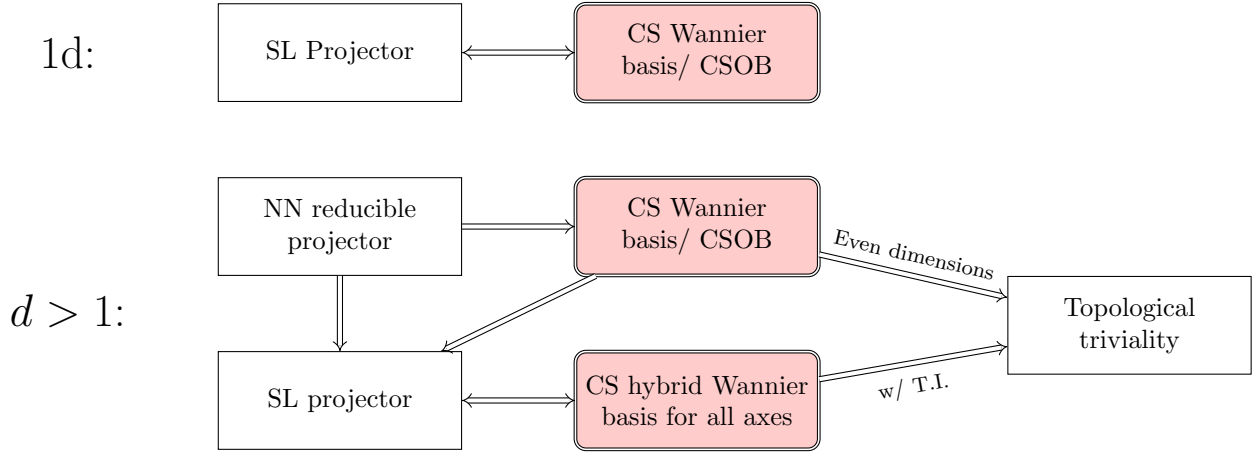

Localization properties of Wannier functions are closely related to the associated bands' topology~\cite{thouless1984wannier, Monaco2018, panati2007triviality, Brouder2007}. Indeed, for many classes of Hamiltonians, exponentially localized Wannier functions  exist iff. the band is topologically trivial. Hence, in addition to the Main Result, we also discuss topological properties of projectors associated with CSOBs and CS hybrid Wannier functions. In particular, using existing results from the literature, we show that all translationally invariant SL projectors in $d>1$ are topologically trivial. Moreover, for even dimensional systems, we show that if a space is spanned by a CSOB, then it is necessarily Chern trivial, irrespective of translational invariance. We also show that if CS hybrid Wannier functions exist for any choice of the localization axis, the corresponding projector is \emph{necessarily} topologically trivial for a translationally invariant system. This contrasts with exponentially localized hybrid Wannier functions, which can be constructed for any band regardless of its topological properties~\cite{marzari1997maximally, Qi2011}. We summarize the Main Results as well as these extra results in figure \ref{fig:results_summary}.

We note that each of the three parts of the Main Result consists of a necessary condition and a sufficient condition for the existence of a CSOB (or a CS hybrid Wannier functions). The necessary condition in all cases is that the projector should be SL, and is straightforward to prove. Proving the sufficient condition is harder, and hence a significant portion of this paper deals with this aspect. Since the sufficient conditions are different for 1d and higher dimensions, these two cases are discussed separately. In section \ref{sec:preliminary}, we discuss the notation and definitions used in this paper, and prove the necessary direction from the Main Result. In addition, we also discuss the relationship between CLSs and CS Wanner-type functions, and show how to obtain non-orthogonal CLSs corresponding to any SL projector. In section \ref{sec:COB for 1d}, we prove by construction for 1d systems, that the image of any SL projector is spanned by a CSOB (or a CS Wannier basis if translationally invariant). Similarly, in section \ref{sec:higher_dimensions}, we prove the sufficient part of points (2) and (3) of the Main Result. In addition we prove the extra results pertaining to topological triviality mentioned above. We conclude the paper with section \ref{sec:conclusion}, and provide a simple method for constructing some SL projectors in Appendix \ref{sec:NN P method}.

\section{Preliminary discussion} \label{sec:preliminary}
In this section, we discuss some basic definitions and notation used in this paper as well as the connections between compact localized states that arise in flat-band systems, and compactly supported Wannier-functions associated with strictly local projectors. We also prove the sufficient part of the Main Result.

\subsection{Notation} 
\label{subsec: notation}
We consider $d$ dimensional tight-binding models, with any (single particle) operator being represented by a matrix with rows and columns labeled by pairs of indices $(\vec{r}, i )$, with $ \vec r \in \z ^d$ denoting a Bravais lattice site position, and $ i \in \{ 1,\dotsc, n_{\vec{r}} \}$ denoting the orbital index (which subsumes all quantum numbers, including spin quantum numbers if present). Our conclusions remain valid for finite lattices as well, for which we replace $\z^d$ by an appropriate set of integer tuples. We denote a position basis vector by $\ket{\vec{r},i}$, and refer to it as orbital $i$ at site $\vec{r}$.

Henceforth, we use the terms site and cell interchangeably. Specifically, a cell at location $\vec{r}$ will mean the same as the site at location $\vec{r}$. By a supercell representation of the lattice, we mean a labelling scheme in which multiple sites in the original lattice representation are grouped together to form a new site. This reversible transformation involves relabeling of quantum numbers, as described in section \ref{subsec:conversion}.

We denote the Hilbert space associated with all the orbitals at cell $\vec{r}$ by $\h^{\vec{r}}$, and the total Hilbert space by $\h_{total}$.  We note that 
\begin{align}
	\h_{total} = \bigoplus_{\vec{r}\in \z^d} \h ^{\vec{r}},\label{eq:Hilbert is direct sum}
\end{align}
with $\oplus$ denoting a direct sum. In general, the number of orbitals at cell $\vec{r}$, denoted by $n_{\vec{r}}$, may be different for different cells, and our conclusions do not depend on them being equal. For notational simplicity, we assume without loss of generality that $n_{\vec{r}}=n$ is independent of the location. For systems with translational invariance, this condition is automatically satisfied. Additionally we find it convenient to rewrite $\mathcal{H}_{total}$ as,
\begin{align}
\h_{total} = \z^{\otimes d } \otimes \h, \label{eq:Hilbert is tensored}
\end{align}
where $\h$ denotes the $n$-dimensional orbital space. We refer to any orthonormal basis vectors of $\h$ as orbitals. 

In this paper, we consider orthogonal projection operators, i.e. operators $P:\h_{total}\rightarrow \h_{total}$, that satisfy $P^2 = P^\dagger = P$, with $(.)^\dagger$ denoting the matrix conjugate transposition operation.
We define a strictly local (SL) projection operator to be one which has a finite upper bound on the extent of its hopping elements:
\begin{definition}
An orthogonal projection operator $P:\h_{total} \rightarrow \h_{total}$ is said to be strictly local if there exists a finite integer $b$ such that $\bra{\vec{r},i}P\ket{\vec{r}\ ',j} = 0 \ \forall \ |\vec{r}-\vec{r}\ '|>b$ and $i,j\in \{1,\dotsc,n\}$. The maximum hopping distance of $P$ is the smallest integer $b$ which satisfies this condition.
\end{definition}

A wavefunction which has non-zero support only on a finite number of sites is said to be compactly supported. Specifically,
\begin{definition}
A wavefunction is compactly supported iff. there exists a finite integer $r$, such that it has zero support outside a ball of radius $r$. The smallest integer value of $r$ is called the size of the wavefunction. A basis is said to be a compactly supported orthogonal basis (CSOB) of (a finite) size $R$ iff. every constituent wavefunction is compactly supported and has a size of at most $R$.
\end{definition}

In the context of finite sized lattices, a basis is considered to be compactly supported only if $R$ is smaller than the size of the lattice. Similarly, only those projection operators that have a maximum hopping distance smaller than the size of the lattice will be considered to be strictly local.

For an SL projector $P$, the following notation will be used in the paper:
\begin{enumerate}
	\item Let $\Pi_{\vec{r}}^P$ denote the set of vectors obtained by operating $P$ on all the orbitals at cell $\vec{r}$. That is,
	\begin{align}
		\p_{\vec{r}}^P \coloneqq \{ P \ket{\vec{r},i} : i \in 1,\dotsc, n \}. \label{eq:defn piz}
	\end{align}
	The choice of which orbital basis is chosen while calculating $\p_{\vec{r}}^P$ will be specified, or will be clear from the context.
	\item Let $\h_{\vec{r}}^P \subset \h_{total}$ denote the space spanned by $\Pi_{\vec{r}}^P$. 
	\item Let $\h^P \equiv \cup_{\vec{r}\in\z^d} \h_{\vec{r}}^P$ denote the image of $P$.
	\item Let $\widetilde{\Pi}_{\vec{r}}^P$ denote an orthonormal basis of $\h_{\vec{r}}^P$. ($\widetilde{\Pi}_{\vec{r}}^P$ can be obtained by applying the Gram-Schmidt orthogonalization procedure on $\Pi_{\vec{r}}^P$). We note that since $P$ is SL, all wavefunctions in the sets $\widetilde{\Pi}_{\vec{r}}^P$ and ${\Pi}_{\vec{r}}^P$ are compactly supported. 
	\item Let $P_{\vec{r}}$ denote the orthogonal projection operator onto $\h_{\vec{r}}^P$. $P_{\vec{r}}$ can be expressed as
		\begin{align}
			P_{\vec{r}} = \sum_{\ket \chi \in \tilde{\Pi}_{\vec{r}}^P} \ket \chi \bra \chi . \label{eq:Pz and pi}
		\end{align}
\end{enumerate}

\subsection{Compactly Supported Wannier functions and Compact Localized States} \label{subsec:construction_nonortho}

While the primary object of interest of this paper is compactly supported (CS) Wannier functions, a closely related type of basis consists of CS Wannier-type functions, which exist prominently in flat-band Hamiltonians. As we will show later, CS Wannier functions' existence is related to the strict localization of the associated projector. While flat-band projectors need not be SL projectors and vice versa, both their images are spanned by CS Wannier-type functions and in some cases CS Wannier functions as well. In this section, we will discuss these two types of bases for flat-band Hamiltonians and SL projectors.

Let us first discuss Wannier-type functions, which are in a sense a generalization of Wannier functions. Similar to Wannier functions, Wannier-type functions consist of a set of wavefunctions localized at a cell, and all their lattice translates, and span a band or a set of bands. However, unlike Wannier functions which are by definition orthogonal, Wannier-type functions can be non-orthogonal, or even linearly dependent. Consequently, a set of $m$ bands may be spanned by $l\geq m$ \emph{flavors} of Wannier-type functions, whereas exactly $m$ flavors of Wannier functions span $m$ bands. Wannier-type functions that are CS~\cite{dubail2015tensor, Read2017} are desirable in certain applications~\cite{ozolicnvs2013compressed}. Importantly, the existence of non-orthogonal CS Wannier-type functions does not in general imply the existence of CS Wannier functions.

CS Wannier-type functions exist most notably as bases spanning flat bands in flat-band Hamiltonians~\cite{chen2014impossibility,leykam2018artificial}. Such functions corresponding to a flat band are also Hamiltonian eigenstates, and are also referred to as compact localized states (CLSs). In most flat-band Hamiltonians, the CLSs are not mutually orthogonal, and can even by linearly dependent. Indeed, in the presence of band touching, CLSs may not even span the entire flat band~\cite{Bergman2008}. However, it is always possible to modify such models so that they have orthogonal CLSs spanning a flat band in an enlarged unit cell. This can be done by choosing a subset of CLSs, comprising regularly spaced CLSs with no physical overlap on the lattice~\cite{Huber2010, Kuno2020b}. While such a set does not span a full band in a primitive cell representation, they always span an entire flat band in an appropriately enlarged unit cell. Hence, they can also be referred to as CS Wannier functions. In some cases, flat-band Hamiltonians possess orthogonal CLSs naturally, without needing unit cell enlargement. Some popular examples from the literature with such bases are discussed in sections \ref{sec:COB for 1d} and \ref{sec:higher_dimensions}. It is possible to create many more examples of flat-band Hamiltonians with orthogonal CLSs, by constructing nearest neighbor projectors as done in Appendix \ref{sec:NN P method}.

\begin{figure}[t]
     \centering
     \begin{subfigure}[b]{0.75\textwidth}
         \centering
         \includegraphics[width=\textwidth]{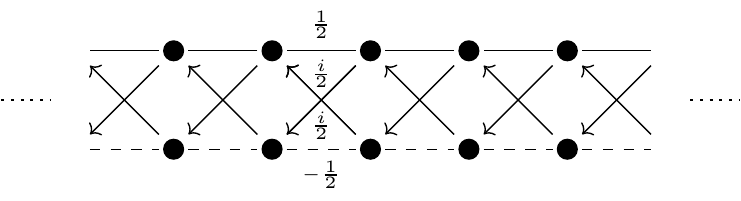}
         \caption{}
         \label{fig:CLS_nonortho_example}
     \end{subfigure} \\
     \begin{subfigure}[b]{0.15\textwidth}
         \centering
         \includegraphics[width=\textwidth]{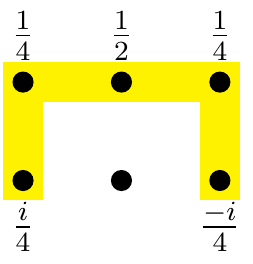}
         \caption{}
         \label{fig:CLS_1}
     \end{subfigure}
	\hspace{2cm}
	\begin{subfigure}[b]{0.15\textwidth}
         \centering
         \includegraphics[width=\textwidth]{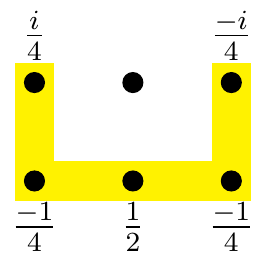}
         \caption{}
         \label{fig:CLS_2}
     \end{subfigure}    
     \hspace{2cm}
	\begin{subfigure}[b]{0.15\textwidth}
         \centering
         \includegraphics[width=\textwidth]{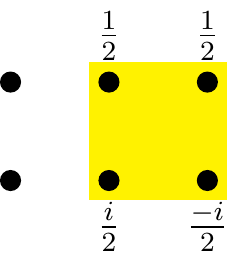}
         \caption{}
         \label{fig:CLS_3}
     \end{subfigure}   
     
            \caption{The real space connectivity of the projector from equation \eqref{eq:1d_simple_P_example} is shown in (a). The upper and lower array of dots represents the two sublattices A and B respectively. An on-site potential of $\frac{1}{2}$ is present for each orbital. Two flavors of non-orthogonal CLSs, i.e. CS Wannier-type functions spanning the image of the projector are shown in (b) and (c). An orthogonal CLS, i.e. a CS Wannier function is shown in (d).}
        \label{fig:1dcls_example}
\end{figure}

While the term CLSs is commonly used only in the context of flat-band Hamiltonians, the distinction between orthogonal CLSs and CS Wannier functions is unnecessary for our purpose. Indeed, it is always possible to deform a band's energy without changing the subspace corresponding to it. Expressing a Hamiltonian as $H(\vec{k}) = \sum_i E_i(\vec{k}) P_i (\vec{k})$, where $P_i$'s are the band projectors, one can modify a band's energy function to `flatten' the corresponding band~\cite{parameswaran2013fractional} without modifying the band projectors, and vice versa. Such an operation does not affect the Wannier and Wannier-type functions spanning that band, since they are associated with the band subspace, and have no dependence on the band dispersion. Thus, (non-)orthogonal CLSs are a special type of CS Wannier(-type) functions. As a result, the conditions associated with the existence of orthogonal CLSs spanning a flat band, and those associated with the existence of CS Wannier functions for any band which may or may not be flat, are equivalent to each other.

These arguments also highlight that the properties of the band projector are connected to the existence of CS Wannier functions or orthogonal CLSs spanning a band. Indeed, as we will show later in the paper, the existence of CS Wannier functions is tied to the band projector being SL. Like SL flat-band Hamiltonians, SL projectors also possess CS Wannier-type functions. In 1d, they have the additional property that their images are spanned by CS Wannier functions. Regardless of the dimension, non-orthogonal CLSs can be constructed straightforwardly for any SL projector. To that end, we first note that for any lattice vector $\ket{\vec{r},\alpha}$, since $P^2=P$, 
\begin{equation*}
P ( P \ket{\Vec{r},\alpha}) = (P \ket{\Vec{r}, \alpha}).
\end{equation*}
Thus, if $P \ket{\Vec{r},\alpha} \neq 0$, it is an eigenvector of $P$ with eigenvalue $1$. Since $P$ is strictly local, $P \ket{\Vec{r},\alpha}$ is compactly supported. Additionally, if $P$ is translationally invariant, then the set $\{  P \ket{\Vec{r},\alpha} |\Vec{r}\in \z^n\}$ is one flavor of CLS. If there are $l$ number of $\alpha$'s for which $P\ket{\Vec{x},\alpha}\neq 0$, we obtain $l$ number of CLSs which together span the band(s) corresponding to $P$. In general, without further processing, none of these wavefunctions are guaranteed to be mutually orthogonal. Moreover, it is possible for two wavefunctions within the same flavor of CLSs to be non-orthogonal to each other.

The existence of CS Wannier-type functions for SL projectors can be understood as a destructive interference phenomenon, similar to CLS in flat-band Hamiltonians. This is based on the rather simple observation that any SL projector can also be regarded as a flat-band Hamiltonian with two flat bands. The CLSs for the band with energy 1 are exactly the CS Wannier-type functions spanning the SL projector's image. Although gapped flat-bands are always spanned by a set of CLSs, they need not be orthogonal CLSs. Indeed, it is impossible to find orthogonal CLSs spanning a flat-band for many flat-band Hamiltonians. Hence, interpreting an SL projector as a flat-band Hamiltonian does not directly help us to conclude that CS Wannier functions spanning it exist. Especially in 1d, this is a reflection of the fact that SL projectors are a subset of flat-band projectors. For some explicit examples of flat-band projectors that are not SL, see section 3 in reference~\cite{chen2014impossibility}.

We illustrate many of these points using a simple example of a 1d projector (see figure~\ref{fig:CLS_nonortho_example}), given by
\begin{align}
P(k) = \frac{1}{2} \begin{pmatrix}
1 + \cos k & \sin k \\
\sin k & 1 - \cos k
\end{pmatrix}. \label{eq:1d_simple_P_example}
\end{align}
We can construct two sets of CLSs by operating the projector on each $A$ and $B$ sublattice orbital (figures \ref{fig:CLS_1} and \ref{fig:CLS_2} show one CLS each from these two sets).  Each set of CLS is not only non-orthogonal, but also linearly dependent, so that neither of the two sets span the image of the projector individually. However, both sets of CLSs considered together span the image of the projector, and hence form two flavors of CS Wannier-type functions. As we will show in the next section, one can always find CS Wannier functions spanning the image of any 1d SL projector. For the example under consideration, this consists of the CLS shown in figure \ref{fig:CLS_3}, and all its lattice translates. It can be easily verified that this set is orthogonal, and hence spans the image of $P$. While it is possible to also obtain these orthogonal CLSs using only destructive interference-based observations, it does not follow immediately that this can be done for projectors with more complicated connectivity, higher number of dimensions (for example, see \eqref{eq:2d_NN_P_example}) or in the absence of lattice translation symmetry.

Thus, while it is clear that for any SL projector, one can construct non-orthogonal CLSs, it is far less obvious (and possibly untrue for $d>1$) that one can construct orthogonal CLSs. One of the objectives of this paper is to present a systematic procedure for the construction of such orthogonal CLSs/CS Wannier functions, and the identification of conditions required for the existence of such functions.

Before proceeding, we note that for non-translationally invariant SL projectors, although CLSs as defined above don't exist, we may define an analogous basis. Specifically, the set $\{ P\ket{\vec{r},\alpha} | \vec{r}\in \z^d, \alpha = 1,\dotsc,n\}$ is a non-orthogonal basis of the image of $P$, and consists of CS wavefunctions.

\subsection{Compact Basis and Strictly Local Projectors} \label{subsec:necessary}
The primary focus of this paper is the identification of necessary and sufficient conditions for the existence of a compactly supported orthogonal basis (CSOB) corresponding to a subspace. In this section, we prove the necessary condition for all dimensions: the strict locality of the associated orthogonal projector.

Proving that the span of a CSOB always corresponds to an SL projector is straightforward. Let a set $S$ be a CSOB of size $R$ on a d-dimensional lattice. Let $P$ be the orthogonal projector onto the space spanned by $S$. Then for any two locations $z,z'\in \z^d$ such that $|z-z'|>R$, and orbitals $\alpha,\beta$, we note that 
 \begin{align*}
     \bra{z,\alpha} P \ket{z',\beta} &= \sum_{\ket{\chi} \in S} \bra{z,\alpha} \ket{\chi}\bra{\chi} \ket{z',\beta} \\
                                     &= 0.
 \end{align*}
In other words, the maximum hopping distance of $P$ is at most $R$. By definition, $P$ is then an SL projector.

Following a similar reasoning, we can easily show that if an orthogonal basis consists of wavefunctions each of which are compactly supported along only one axis, then the corresponding projector is strictly local along that axis. This implies that if CS hybrid Wannier functions exist for all choices of the localized axis, then the projector is SL (along all directions). This proves the necessary part of point (2) of the Main Result.

\section{Compactly Supported Orthogonal Basis: 1d lattices} \label{sec:COB for 1d}
\begin{figure}
      \centering
      \begin{subfigure}[b]{0.40\textwidth}
          \centering
          \includegraphics[width=\textwidth]{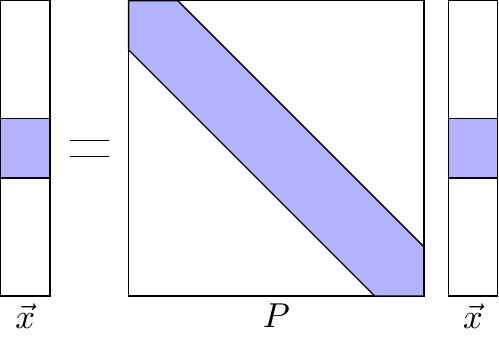}
          \caption{}
          \label{fig:band_proj}
      \end{subfigure} \hspace{0.15\textwidth}  
      \begin{subfigure}[b]{0.40\textwidth}
          \centering
          \includegraphics[width=\textwidth]{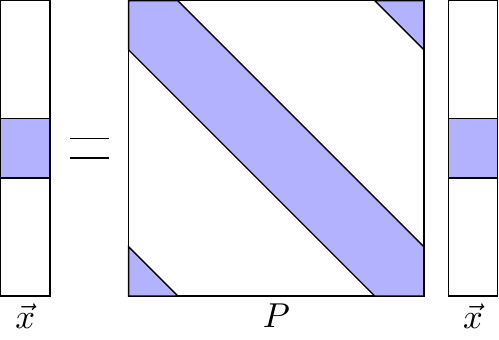}
          \caption{}
          \label{fig:band_proj_periodic}
 	\end{subfigure}
             \caption{(a) SL projectors on 1d lattices are represented by band diagonal projection matrices, and CS wavefunctions are represented by sparse vectors which have non-zero components only within a finite patch (shown in blue). The main result for 1d projectors implies an orthogonal projection operator is band diagonal if and only if it possess a CS orthogonal eigenbasis. (b) 1d SL projectors on finite periodic lattices are band diagonal, with appropriate modifications at the corners. In theorem \ref{thm:1d}, a relation between the band width of the projector and the number of non-zero elements of the basis vectors is provided.}
         \label{fig:band_proj_statement}
\end{figure}

As shown in section \ref{subsec:necessary}, it is straightforward to show that if a compactly supported orthogonal basis (CSOB) or a CS Wannier basis exists, then the corresponding projector is necessarily SL. In this section, we prove the converse, i.e. if a 1d projector is SL, then there exists a CSOB spanning its image, thereby completing the proof of part (1) of the Main Result. Specifically, we will show that
\begin{theorem}\label{thm:1d}
In 1d systems, the span of a set of occupied states possesses a compactly supported orthogonal basis (CSOB) if and only if the orthogonal projector onto the span is strictly local. Additionally:
\begin{enumerate}[(1)]
\item If the maximum hopping distance of an SL projector is $b$, then there exists such a basis consisting of wavefunctions of a maximum size of $3b$ cells, irrespective of the presence of translational invariance.
 \item If the projector is translationally invariant, its image is spanned by a compactly supported Wannier basis in a size $2b$ supercell representation of the lattice.
\end{enumerate}
\end{theorem}
We prove these statements through algorithmic constructions of CSOBs and CS Wannier functions, and provide bounds on their sizes on the way.

Let us briefly discuss some properties of SL projectors and outline the approach we will use in this section. We note that 1d SL projectors are band-diagonal matrices when expressed in the orbital basis (see figure~\ref{fig:band_proj_statement}). For example, every 1d nearest neighbor (NN) projector can be represented by a block tridiagonal matrix, with the block size being equal to the number of orbitals per cell. Although it is straightforward to obtain a CSOB for a block size of one, the corresponding statement is not obvious for larger block sizes. However, using the Gram-Schmidt orthogonalization procedure with an appropriate orthogonalization sequence, we show that it is always possible to obtain such an eigenbasis for any block size (see section \ref{subsec:GS}). If the projector is also translationally invariant, i.e. with repeating blocks in the matrix representation, the basis can be chosen to be a Wannier basis in a supercell representation (see section~\ref{subsec:Wannier_1d}). We then extend these results to all 1d SL projectors, since they can be represented as NN projectors using supercell representations (see section \ref{subsec:conversion}). In terms of matrices, this is equivalent to expressing any band diagonal matrix as a block-tridiagonal matrix by grouping together the original blocks into appropriate larger blocks. In this sense, the problem of obtaining a CSOB for any SL projector is equivalent to the problem of obtaining one for an NN projector.

\begin{figure}[b]
    \centering
        \includegraphics[width=0.65\textwidth]{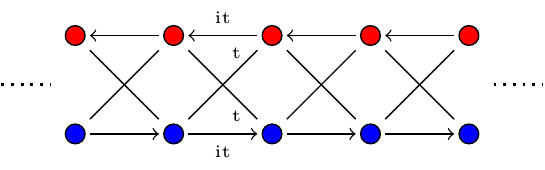}
        \caption{Example of a popular model with CS Wannier functions: the Creutz ladder. The hopping amplitudes result in two flat bands and CS Wannier functions. The A (B) orbitals are shown in red (blue). The hopping amplitudes are shown next to the arrows.}
        \label{fig:Creutz_ladder}
\end{figure}
While a number of model Hamiltonian systems possessing CS Wannier functions and orthogonal CLSs have been studied in the literature, to our knowledge, this is the first time that the connection of such systems with the strict localization of the corresponding projection operator has been explicitly established. Most examples from the literature with CS Wannier functions involve flat-band Hamiltonians defined on the Creutz ladder~\cite{Creutz2001}, the sawtooth lattice~\cite{Huber2010,Kuno2020b} and the diamond lattice~\cite{Leykam2013}. Here, we briefly discuss the Creutz lattice, which has been extensively studied both theoretically~\cite{Kuno2020,Junemann2017, Creutz1999} and experimentally~\cite{Mukherjee2018, Kang2020}. As noted in~\cite{Creutz1999}, for one choice of parameters (see figure \ref{fig:Creutz_ladder}), the Creutz ladder has two exactly flat bands of energies $\pm 2t$, each of which is spanned by CS Wannier functions. The corresponding Hamiltonian in k-space is given by:
\begin{align}
H(k) = 2t\begin{pmatrix}
\sin k & \cos k\\
\cos k & -\sin k
\end{pmatrix}.
\end{align}
As mentioned in section \ref{subsec:construction_nonortho}, SL flat-band Hamiltonians possess some destructive interference properties, which constrain the movement of initially localized particles. For example, in the system under consideration, a particle initially localized on an A orbital cannot diffuse to a B orbital located more than a hop away (see figure 2 and related discussion in reference~\cite{Creutz1999}). Based on this observation, one can obtain orthogonal CLSs, or equivalently, CS Wannier functions for the two bands of this Hamiltonian. The Wannier functions (labeled by $\pm$ for the two bands) localized at cell $z$ are given by:
\begin{align*}
\ket{W_\pm} &= \frac{1}{2} \left( \pm i\ket{z,A} \pm \ket{z,B}+ \ket{z+1,A} + i \ket{z+1,B}\right).
\end{align*}
In accordance with the predictions of our paper, the band projectors onto the two bands are SL, and are given by
\begin{align}
P_\pm(k) = \frac{1}{2} \begin{pmatrix}
1\pm \sin k &\pm \cos k \\
\pm \cos k & 1\mp \sin k
\end{pmatrix}.
\end{align}
Indeed, using the techniques developed in the next subsection, one can obtain these CS Wannier functions from the expressions for the projectors $P_\pm(k)$.

While most popular flat-band Hamiltonians do not possess CS Wannier functions, i.e. a set of orthogonal CLSs spanning the flat band, as discussed in section \ref{subsec:construction_nonortho}, it is possible to construct models with orthogonal CLSs by enlargement of the unit cell of any known flat-band model. This was done for example in \cite{Huber2010} and~\cite{Kuno2020b} to obtain CS Wannier functions spanning a flat band in the sawtooth lattice.

Although most of the examples from the literature involve flat bands, CS Wannier functions and SL projectors can correspond to dispersive bands. Such an example is provided in Appendix~\ref{subsec:1d_example}, along with a discussion of the CS Wannier functions obtained using the algorithm from the next subsection.

\subsection{Nearest Neighbor Projectors} \label{subsec:GS}
In this section, we present an algorithm for obtaining a CSOB corresponding to any nearest neighbor (NN) projector. The algorithm is based on the Gram-Schmidt orthogonalization procedure, and produces a CSOB with each basis vector having a maximum spatial extent of $3$ consecutive lattice cells.  The methods in this section are applicable even if the projector is not translationally invariant.

The basic idea underlying our procedure is to obtain the set $\widetilde{\Pi}_z^P$ defined in section \ref{subsec: notation}, corresponding to the localized eigenstates of $P_z$ for some lattice site $z$, and to `reduce' $P$ to $P-P_z$ as required by the Gram-Schmidt procedure. Then, we operate this reduced projector $P-P_z$ on another cell $z'$, and iterate along a sequence of cell locations which includes all the integers. The union of all the $\widetilde{\Pi}_{z}^P$ sets will form an orthonormal basis for the image of $P$. While at the first step, $\widetilde{\Pi}_z^P$ is guaranteed to have compactly supported wavefunctions, it is not obvious that the size of the corresponding wavefunctions stays bounded for subsequent steps. With the help of the following lemma, we can show that the size of each vector at any step will be at most three consecutive cells.

\begin{lemma}
\label{lemma: Pz left-right}
For any $z\in\z$, orbital indices $i,j \in \{1,\dotsc,n\}$, and integers $\delta,\delta'>0$,
\begin{align}
\bra{z-\delta, i} (P-P_z) \ket{z+\delta', j} = 0. \label{eq:Pz left-right connections}
\end{align}

Additionally,
\begin{align}
(P-P_z) \ket{z,j} = 0. \label{eq:Pz z no connections}
\end{align}

\end{lemma}

\begin{figure}[b]
     \centering
     \begin{subfigure}[b]{0.55\textwidth}
         \centering
         \includegraphics[width=\textwidth]{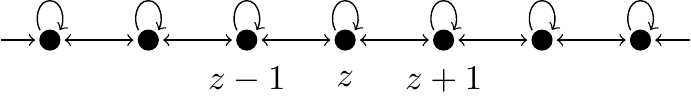}
         \caption{}
         \label{fig:connectivity_before}
     \end{subfigure} \\
     \begin{subfigure}[b]{0.55\textwidth}
         \centering
         \includegraphics[width=\textwidth]{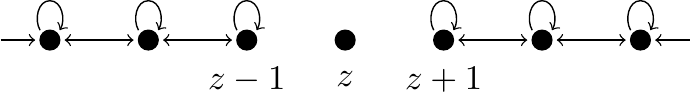}
         \caption{}
         \label{fig:connectivity_after}
     \end{subfigure}
            \caption{Change in the connectivity of an NN projector after a Gram-Schmidt step. Each dot represents a cell, with arrows indicating possible non-zero matrix elements of $P$. (a) Connectivity of an NN projector $P$. (b) Connectivity of the reduced projector $P-P_z$. Missing arrows indicate that the corresponding matrix elements are zero.}
        \label{fig:connectivity}
\end{figure}

\begin{proof}
For any cell $z$, let $P_{zz}:\h^z \rightarrow \h^z$ denote the $n\times n$ matrix corresponding to the matrix elements of $P$ between different orbitals at cell $z$, i.e. 
\begin{align*}
\left( P_{zz} \right) _{ij} = \bra{z,i} P \ket{z,j}.
\end{align*}
Since $P$ is Hermitian, $P_{zz}$ is also Hermitian. Hence, there exists a unitary matrix $U_z$, such that $D_z = U_z^\dagger P_{zz} U_z$ is a real diagonal matrix. Let $d_z$ denote the vector of diagonal elements of the matrix $D_z$. The unitary $U_z$ defines a new basis at $z$:
\begin{align}
    \ket{z,i}' \coloneqq \sum_j (U_{z})_{ji} \ket{z,j}. \label{eq:diagonal basis}
\end{align}
We call the orbital basis $\{ \ket{z,\alpha}' | \alpha=1,\dotsc,n\}$ a diagonal basis at $z$. Since $\bra{z,i}' P \ket{z,j}' = d_i \delta_{ij}$, we can significantly reduce the intra-cell connectivity of $P$, by performing the unitary transformation at every cell separately. Equivalently, we use the global unitary transformation
\begin{align}
    U \coloneqq \mathop{\bigoplus}_{z\in \z} U_z.  \label{eq:diagonal_basis_1d}
\end{align}
Hereafter, for notational convenience, we drop the prime outside the vectors, and assume that we have already rotated the basis to a diagonal one. Thus, $\bra{z,\alpha}P\ket{z,\beta} = 0$ whenever $\alpha\neq \beta$. \par 

Now we obtain a few important identities using the fact that $P$ is an orthogonal NN projector, and using appropriate insertions of resolution of identity:
\begin{enumerate}
	\item Since $P$ is positive-semidefinite, for any cell $z$ and orbital $\alpha$, we have 
    \begin{align}
       \bra{z,\alpha} P \ket{z,\alpha} &\geq 0.   \label{eq:1d_self_hop_positive_def}
    \end{align} 

    \item For any two neighboring cells $z$ and $z+1$ and orbitals $\alpha, \beta$, we have 
    \begin{align}
         \bra{z,\alpha}P\ket{z,\alpha} &+ \bra{z+1,\beta}P\ket{z+1,\beta} = 1, \quad 
        \text{ whenever }\bra{z,\alpha}P\ket{z+1,\beta} \neq 0.  \label{eq:1D_sum_of_hoppings_is_1}
    \end{align}
    This follows from
    \begin{align*}
        \bra{z,\alpha}P\ket{z+1,\beta} &= \bra{z,\alpha}P^2\ket{z+1,\beta}  \\
        & = \bra{z,\alpha}P\ket{z,\alpha}\bra{z,\alpha}P\ket{z+1,\beta} + \bra{z,\alpha}P\ket{z+1,\beta} \bra{z+1,\beta}P\ket{z+1,\beta}  \\
        &= \bra{z,\alpha}P\ket{z+1,\beta} \left( \bra{z,\alpha}P\ket{z,\alpha} + \bra{z+1,\beta}P\ket{z+1,\beta}\right). 
    \end{align*}
    
    \item Similarly, for any two cells $z$ and $z+2$ separated by two hops, since $P$ is an NN operator,
    \begin{align}
        \sum _ \gamma \bra{z,\alpha} P \ket{z+1,\gamma} \bra{z+1,\gamma} P \ket{z+2,\beta} = 0. \label{eq:1D_sum_prod_is_zero}
    \end{align}
\end{enumerate}
Now we obtain the projector $P_z$ by orthogonalizing $\p_z$ (using the diagonal basis orbitals of $P_{zz}$ in expression \ref{eq:defn piz}). 
Since $\bra{z,\alpha}P\ket{z,\beta}=0$ whenever $\alpha\neq \beta$, $\p_z$ is already orthogonal, so we only need to normalize the vectors in it in order to obtain an orthonormal set.
Whenever $P \ket{z,\alpha} \neq 0$, we denote the corresponding normalized vector by 
\begin{align*}
    \ket{P,z,\alpha} \coloneqq \frac{P\ket{z,\alpha}}{\sqrt{\bra{z,\alpha}P\ket{z,\alpha}}}.
\end{align*}{}
Thus, we obtain $P_z$ by adding the projector onto each orthonormal vector:
\begin{align*}
    P_z &\equiv \sum_\alpha^{(d_z)_{\alpha} \neq 0} \ket{P,z,\alpha}\bra{P,z,\alpha} = \sum_\alpha^{(d_z)_{\alpha}\neq 0} P\frac{\ket{z,\alpha}\bra{z,\alpha}}{\bra{z,\alpha}P\ket{z,\alpha}} P.
\end{align*}
Since $\p_z$ consists of wavefunctions with non-zero support only on cells $z$ and $z\pm1$, equation \eqref{eq:Pz left-right connections} is already satisfied, whenever $\delta>1$ or $\delta'>1$. Thus, we only need to verify that \eqref{eq:Pz left-right connections} is satisfied for $\delta = \delta ' = 1$, for which, we get
\begin{align*}
    \bra{z-1,\alpha} P_z \ket{z+1,\beta} &= \sum_\gamma^{(d_z)_{\gamma} \neq 0} \frac{\bra{z-1,\alpha}P\ket{z,\gamma}\bra{z,\gamma}P\ket{z+1,\beta}}{\bra{z,\gamma}P\ket{z,\gamma}}.
\end{align*}

Any non-zero term in the summation will have $\bra{z-1,\alpha}P\ket{z,\gamma} \neq 0$. From condition \eqref{eq:1D_sum_of_hoppings_is_1}, all such orbitals $\gamma$ at cell $z$ possess the same self-hop:
\begin{align*}
    \bra{z,\gamma}P\ket{z,\gamma} = 1- \bra{z-1,\alpha} P \ket{z-1,\alpha} \neq 0.
\end{align*}
Thus, we get 
\begin{align}
    \bra{z-1,\alpha} P_z \ket{z+1,\beta} &= \frac{1}{1- \bra{z-1,\alpha} P \ket{z-1,\alpha}}\sum_\gamma \bra{z-1,\alpha}P\ket{z,\gamma}\bra{z,\gamma}P\ket{z+1,\beta} \nonumber\\
    &= 0, \label{eq:Px_connectivity}
\end{align}
where we have used condition \eqref{eq:1D_sum_prod_is_zero}. Since $P$ has vanishing matrix elements between orbitals lying on opposite sides of $z$, this proves equation \eqref{eq:Pz left-right connections}.

We note that $P\ket{z,\alpha} = P_z \ket{z,\alpha}$ and hence $(P-P_z) \ket{z,\alpha} = 0$. This leads to equation \eqref{eq:Pz z no connections}. The connectivity of the reduced projector is shown in figure \ref{fig:connectivity}.
\end{proof}

Based on this lemma, we present a method for obtaining a CSOB  for the image of an NN projector, as described in Procedure \ref{algo:Gram Schmidt procedure}.  

\begin{algorithm}[]
     \caption{Procedure for Constructing a CSOB for the Image of any 1d NN Projector} \label{algo:Gram Schmidt procedure}
    \SetKwInOut{Input}{Input}
    \SetKwInOut{Output}{Output}
    \Input{An NN projector $\mathcal{P}$ acting on a 1d lattice.}
    \emph{Procedure:} Define a non-repeating sequence $S \equiv z_0,z_1,\dotsc$ of integers, s.t. it contains all the integers. Set $P\leftarrow\mathcal{P}$. Initialize $k=0$, and do:
    \begin{enumerate}
    	\item Set $z \leftarrow z_k$.
        \item Obtain $\Pi_z^P$ (as defined in eq. \eqref{eq:defn piz}).
        \item Orthonormalize the set $\Pi_z^P$ to obtain $\widetilde{\Pi}_z^P$.
        \item Obtain $P_z$ from $\widetilde{\Pi}_z^P$ using \eqref{eq:Pz and pi}.
        \item Update $P\leftarrow P-P_z$, increment $k$, and go back to step (1).
    \end{enumerate}
    
    \Output{The set $\widetilde{\Pi} \coloneqq \cup_{z\in \z} \widetilde{\Pi}_z^P$ of compactly supported wavefunctions, which is an orthonormal basis of the image of the projector $\mathcal{P}$.}    
\end{algorithm}

\begin{lemma}
The set $\widetilde{\Pi}$ obtained from the Gram-Schmidt algorithm \ref{algo:Gram Schmidt procedure} is an orthonormal basis of the image of the NN projector $\mathcal{P}$. Furthermore, every element of $\widetilde{\Pi}$ is compactly supported, with a spatial extent of no more than three consecutive cells.
\end{lemma}
\begin{proof}
The set $\cup_z \p_z^P$ spans the image of $\mathcal{P}$. Hence, procedure \ref{algo:Gram Schmidt procedure}, which is the application of the Gram-Schimidt procedure on it, creates an orthonormal basis of the image of $\mathcal{P}$. Lemma \ref{lemma: Pz left-right} implies that the reduced projector $P-P_z$ obtained at any iteration in the procedure is also an NN projector. Thus, every vector belonging to $\widetilde{\p}_z^P$ for any $z$ is guaranteed to be compactly supported, with a maximum spatial extent of $3$ consecutive cells.
\end{proof}
It is possible to choose a sequence $S$ so that at most $n$ of the created basis vectors have a spatial extent of $3$ cells, with all the remaining vectors having a spatial extent of at most $2$ consecutive cells. For an example, see figure \ref{fig:GS_example}.

\subsection{Translationally Invariant Nearest Neighbor Projectors} \label{subsec:Wannier_1d}
Although Procedure \ref{algo:Gram Schmidt procedure} from the previous section also works for translationally invariant projectors, in general, the obtained basis may not consist of Wannier functions. In this section, we will show that it is possible to obtain a Wannier basis consisting of compactly supported functions for the image of $P$, in a size $2$ supercell lattice representation. By a size $2$ supercell representation of the lattice, we mean relabelling the cells so that the lattice is regarded as consisting of ``supercells" which are each twice the size of the original cell. In this representation, the unit cell is the supercell which consists of two primitive unit cells. We also use a supercell representation for the conversion of an SL projector to an NN projector, as will be discussed in the next subsection.
To that end, we first divide the lattice into two subsets, $A$, and $B$, consisting of alternating cells. For concreteness, we choose $A$ to consist of even locations ($2\z$), and $B$ to be the odd locations ($2\z+1$).

We define the set $\p_A^P$ as being
\begin{align*}
\p_{A}^P = \bigcup_{i\in A} \p_i^P.
\end{align*}
We also define $\h_A^P$ to be the span of $\p_A^P$, and $P_A$ to be the orthogonal projection onto $\h_A^P$. From lemma \ref{lemma: Pz left-right}, we note that $P-P_A$ has a significantly reduced connectivity, as shown in figure \ref{fig:connectivity_A}. Specifically, the only non-zero matrix elements of $P-P_A$ between any two orbitals, are those between any two orbitals located at the same cell belonging to set $B$.

\begin{figure}[t]
    \centering
        \centering
        \includegraphics[scale=1]{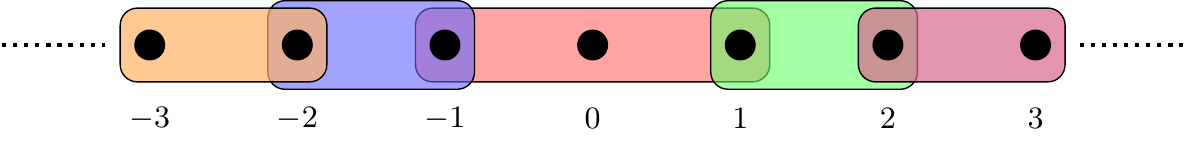}
        \caption{If we use the sequence $S=0,-1,1,-2,2,\dotsc$ in procedure \ref{algo:Gram Schmidt procedure}, then all sets $\widetilde{\p}_z^P$ except for $\widetilde{\p}_0^P$ consist of wavefunctions of a maximum size of $2$. Each colored rectangle represents the maximum spatial extent of the wavefunctions in $\widetilde{\p}_z^P$ obtained during one iteration of the procedure. Each unit cell is represented by a black dot.}
        \label{fig:GS_example}
\end{figure}

\begin{figure}[t]
     \centering
     \begin{subfigure}[b]{0.55\textwidth}
         \centering
         \includegraphics[width=\textwidth]{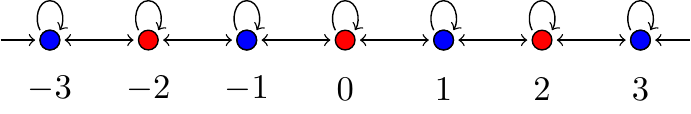}
         \caption{}
         \label{fig:connectivity_before_A}
     \end{subfigure} \\
     \begin{subfigure}[b]{0.55\textwidth}
         \centering
         \includegraphics[width=\textwidth]{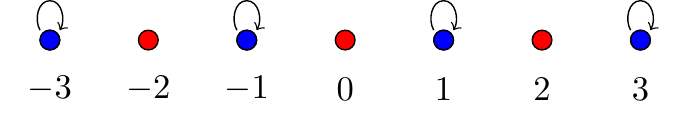}
         \caption{}
         \label{fig:connectivity_after_A}
     \end{subfigure}
            \caption{Sets A and B are shown in red and blue respectively. (a) Connectivity of an NN projector $P$. (b) Connectivity of $P-P_A$. Missing arrows indicate vanishing matrix elements.}
        \label{fig:connectivity_A}
\end{figure}

\begin{algorithm}[h]
    \caption{Compactly Supported Wannier Basis for 1d translationally invariant NN projectors} \label{algo:GS_1d_Wannier}
    \SetKwInOut{Input}{Input}
    \SetKwInOut{Output}{Output}
    \Input{A 1d translationally invariant NN projection operator $P$.}
    \emph{Procedure:} 
    \begin{enumerate}
    	\item Obtain $\p_0^P$, and orthogonalize it to obtain the set $\widetilde{\p}_0^P$.
    	\item Obtain the orthogonal projection operator $P_0$ onto the span of $\p_0^P$.
    	\item Obtain a reduced projection operator $P' \coloneqq P - P_0 - \y^2 P_0 \y ^{\dagger 2} $.
    	\item Obtain and orthogonalize $\p_1^{P'}$ to obtain the set $\widetilde{\p}_1^{P'}$.
    	\item Obtain the set $\tilde{\Pi}$, defined as 
    	\begin{align}
    		\widetilde{\p} = \left( \bigcup_{z\in \z} \{ \y^{2z} \ket{\chi}:\ket{\chi}\in \widetilde{\p}_0^P \} \right) \cup \left( \bigcup_{z\in \z} \{ \y^{2z} \ket{\chi}:\ket{\chi}\in \widetilde{\p}_1^{P'} \}  \right). \label{eq:wannier_1d_union}
    	\end{align}
    \end{enumerate}
    
    \Output{The set $\tilde{\Pi}$ consisting of compactly supported Wannier functions spanning the image of $P$, corresponding to a size $2$ supercell representation.}    
\end{algorithm}

Since $P$ is translationally invariant, it is useful to define a translation operator $\y$, which satisfies
\begin{align*}
\y \ket{z,i} = \ket{z+1,i},
\end{align*}
for all $z\in \z$ and $i\in\{1,\dotsc,n\}$. Since $P$ is translationally invariant, $P = \y^\dagger P \y$. \par 

The algorithm for obtaining a CSOB is summarized in Procedure \ref{algo:GS_1d_Wannier}.

\begin{lemma} \label{thm:Wannier_for_1d}
The output obtained using Procedure \ref{algo:GS_1d_Wannier} is a compactly supported Wannier basis spanning the image of $P$ corresponding to a size $2$ supercell representation.
\end{lemma}
\begin{figure}[h]
    \centering
    \begin{subfigure}[c]{\textwidth}
        \centering
        \includegraphics[scale=1]{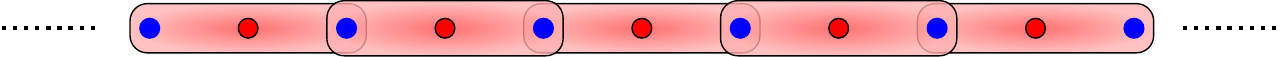}
        \caption{}
        \label{fig:1DchainSublattices}
    \end{subfigure}%
    \vspace{10pt}
    \begin{subfigure}[c]{1\textwidth}
        \centering
        \includegraphics[scale=1]{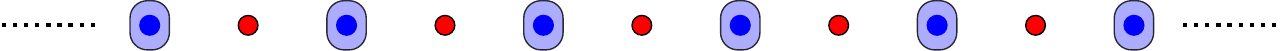}
        \caption{}
        \label{fig:1DchainPPrimeConnectivity}
    \end{subfigure}
    
    \begin{subfigure}[c]{\textwidth}
        \centering
        \includegraphics[scale=1]{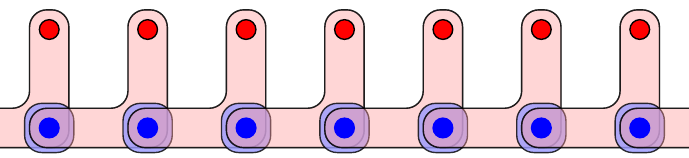}
        \caption{}
        \label{fig:1Dchain_superlattice_Wannier}
    \end{subfigure}
    \caption{The red and blue dots denote the lattice sites belonging to sets $A$ and $B$ respectively. (a) Wavefunctions in $\widetilde{\p}_z^P$ centered at cell $z\in A$ are chosen to be the wavefunctions in $\widetilde{\p}_0^P$ translated by $z$ cells. Each red rectangle centered at $z$ denotes the maximum spatial extent of wavefunctions belonging to $\widetilde{\p}_z^P$. (b) Each blue bubble denotes the maximum spatial extent of the wavefunctions in the set $\widetilde{\p}_{1}^{P-P_A}$, and its translates by an even number of cells. (c) The two sets of functions together form a Wannier basis in a size $2$ supercell representation, with each supercell consisting of one cell each from A and B. }
    \label{fig:1D_alternate_procedure}
\end{figure}
\begin{proof}
The output of Procedure \ref{algo:GS_1d_Wannier} is the set $\widetilde{\p}$, which is a union of two disjoint sets (see equation \eqref{eq:wannier_1d_union}). The first set is a union of the set $\widetilde{\p}_0^P$, and all its even unit cell translates. We will now show that this set is an orthogonal basis of $\h_A^P$.\par
First, we note that the the set $\{ \y^{2z} \ket{\chi} : \ket{\chi} \in \widetilde{\p}_0^P\}$ is an orthogonal basis of $\h_{2z}^P$, since $P$ is translationally invariant. Moreover, for any $z\neq0$, this set is also orthogonal to the set $\widetilde{\p}_0^P$. To see this, we show that \textit{any} orthonormal bases $\widetilde{\p}_{z_1}^P$ and $\widetilde{\p}_{z_2}^P$ for distinct locations $z_1,z_2\in A$ are mutually orthogonal. Let $\ket{\Psi} \in \h_{z_1}^P$ and $\ket{\Phi} \in \h_{z_2}^P$ be two vectors. There exist vectors $\ket{\psi} \in \h^{z_1}$ and $\ket{\phi}\in \h^{z_2}$ such that $\ket{\Psi} = P \ket{\psi}$ and $\ket{\Phi} = P \ket{\phi}$. Taking the inner product of $\ket{\Psi}$ and $\ket{\Phi}$, we obtain
\begin{align*}
\bra{\Phi} \ket{\Psi} &= \bra{\phi} P^\dagger P \ket{\psi} \\
&= \bra{\phi} P \ket{\psi} \\
&= 0,
\end{align*}
since $z_1$ and $z_2$ are located at least two hops away, which is larger than the maximum hopping distance of $P$. Thus, $\h_{z_1}^P$ and $\h_{z_2}^P$ are mutually orthogonal for distinct $z_1,z_2 \in A$. Thus, the set $ \bigcup_{z\in \z} \{ \y^{2z} \ket{\chi}:\ket{\chi}\in \widetilde{\p}_0^P \}$ is an orthogonal basis of $\h_A^P$. Additionally, since $P$ is NN hopping, it consists of compactly supported wavefunctions with a maximum spatial extent of $3$ cells, as shown in figure \ref{fig:1DchainSublattices}. \par 
The second set in equation \eqref{eq:wannier_1d_union} is an orthonormal basis of $\h^P \setminus \h_A^P$. To prove this, we first note that 
\begin{align*}
P' \ket{1,i} &= (P-  P_0 - \y^2 P_0 \y^{\dagger 2}) \ket{1,i} \\
&= \left[ P - \left( \sum_{z\in \z} \y^{2z} P_0 \y^{\dagger 2z} \right) \right] \ket{1,i} \\
&= ( P - P_A ) \ket{1,i}.
\end{align*} 
Hence, the set $\widetilde{\p}_1^{P'}$ is an orthonormal basis of $\h_1^{P-P_A}$. $P-P_A$ remains invariant under translations by an even number of cells. Since $P-P_A$ is also a nearest neighbor projector (using lemma \ref{lemma: Pz left-right}), using the same arguments as for the first set, we conclude that $\bigcup_{z\in \z} \{ \y^{2z} \ket{\chi}:\ket{\chi}\in \widetilde{\p}_1^{P'} \} $ is an orthogonal basis of $\h^P \setminus \h ^P_A$. Additionally, using lemma \ref{lemma: Pz left-right}, we infer that every wavefunction it contains is compactly supported, with non-zero support on only one cell (see figure \ref{fig:1DchainPPrimeConnectivity}). \par 
Thus, $\widetilde{\p}$ is a compactly supported orthogonal basis of $\h^P$. By construction,
\begin{align*}
\y^{2z} \ket{\chi} \in \p \text{ if } \ket{\chi} \in \widetilde{\p},
\end{align*}
for any $z\in\z$. Thus, $\widetilde{\p}$ consists of compactly supported Wannier basis, within a size $2$ supercell representation (figure \ref{fig:1Dchain_superlattice_Wannier}).
\end{proof}

\subsection{Supercell Representation and Strictly Local Projectors} \label{subsec:conversion}

\begin{figure}[ht]
    \centering
    \includegraphics[scale=1]{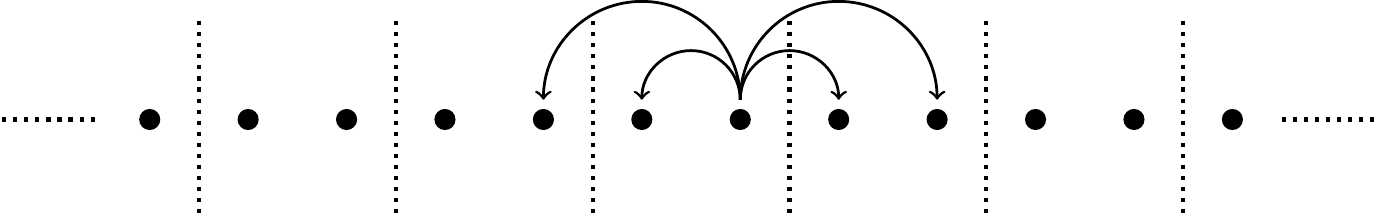}
    \caption{ An example of conversion of a 1d SL operator to an NN operator using a supercell representation. If an operator has a maximum hopping distance of $2$, then grouping the sites in pairs converts the operator to an NN operator in the `supercell' representation.}
    \label{fig:1Dregrouping}
\end{figure}    

As discussed at the beginning of this section, an SL operator on a 1d lattice with a maximum hopping distance $b$ can be expressed as an operator with only nearest neighbor hopping terms using a supercell representation with each supercell consisting of $b$ number of primitive cells (for an illustrative example, see figure \ref{fig:1Dregrouping}). The range of the operator is reduced, at the cost of an increase in the number of orbitals per cell ($n\rightarrow nb$). This transformation enables us to apply the techniques and results for NN projectors ($b=1$) to SL projectors ($b\geq 1$). 

In particular, if we choose the supercell located at the origin to consist of primitive cells $0,1,\dotsc,b-1$, the position and orbital indices in the two representations have the following correspondence: 
\begin{align}
\begin{split}
\text{Primitive cell} &\longleftrightarrow \text{Supercell} \\
\ket{z,i} & \ \equiv \ \  \ket{z \setminus b, n \times (z\bmod b)+i} _s,  
\end{split} \label{eq:Corresp_super_primitive}
\end{align}
with the subscript $s$ denoting a vector in the supercell representation, and $\setminus$ denoting the quotient upon division.  \par

Putting together the results from the previous subsections and the conversion of an SL projector to an NN projector using the supercell representation, we arrive at the following results for arbitrary range SL projectors.
\begin{corollary}
The image of an SL projector on a 1d lattice with a maximum hopping distance $b$ possess a CSOB consisting of wavefunctions of a maximum spatial extent of $3b$ consecutive cells. 
\end{corollary}
We first create a size $b$ supercell representation where $P$ is a NN projector. Using Procedure \ref{algo:Gram Schmidt procedure}, we obtain a CSOB in this supercell representation. We revert back to the original, or primitive cell representation using the correspondence \eqref{eq:Corresp_super_primitive}. This process is summarized in the following sequence:

\begin{center}
\begin{tabular}{ c @{$\longrightarrow$} c @{$\xrightarrow{\text{Procedure \ref{algo:Gram Schmidt procedure}}}$} c @{$\longrightarrow$} c}

\vtop{\hbox{\strut Primitive cell}\hbox{\strut representation}}& \multicolumn{2}{c @{$\longrightarrow$}}{Supercell representation} & \vtop{\hbox{\strut Primitive cell}\hbox{\strut representation}} \\
SL Projector ($b\geq 1$) \quad & \quad NN Projector ($b=1$) \quad &\quad CSOB (max size $3$) \quad  & \quad CSOB (max size $3b$).

\end{tabular}
\end{center}

Similarly, we can use Procedure \ref{algo:GS_1d_Wannier} for obtaining a Wannier basis for arbitrary range translationally invariant SL projectors.

\begin{corollary} \label{corollary:1d_Wannier}
The image of a translationally invariant strictly local projector with a maximum hopping distance $b$ on a 1d lattice is spanned by a compactly supported Wannier basis within a size $2b$ supercell lattice representation.
\end{corollary}
The procedure for obtaining such a basis is summarized in the following sequence:
\begin{center}
\begin{tabular}{c c c c c}
Primitive cell  & $\xrightarrow{\text{Size b}}$ & Size $b$ supercell  & $\xrightarrow{\text{Size 2}}$ & Size $2b$ supercell \\
representation & &  representation & & representation \\
SL Projector ($b\geq 1$) & $\longrightarrow$ &  NN Projector ($b=1$) &  $\xrightarrow{\text{Procedure \ref{algo:GS_1d_Wannier}}}$ & CS Wannier basis.
\end{tabular}
\end{center}
To summarize, we have shown that the image of a strictly local projector in 1d is always spanned by a compactly supported orthogonal basis (or  a compactly supported Wannier basis if the projector has lattice translational invariance). This completes our proof of theorem \ref{thm:1d}.

Having presented a technique for the construction of CS Wannier functions for 1d SL projectors, we apply this technique to an example Hamiltonian with a band associated with an SL projector in Appendix~\ref{subsec:1d_example}. Furthermore, we discuss why the Gram-Schmidt orthogonalization procedure in our algorithm is easier to use instead of the symmetric orthogonalization procedure~\cite{lowdin1950non}. We also compare our results with those obtained using the maximally localized Wannier functions procedure~\cite{marzari1997maximally}.

In the next section we will study how these results can be extended to higher dimensional lattices which have a larger coordination number. We close this section with some comments on the Bethe lattice which is sometimes regarded as an infinite dimensional lattice. For an NN projector on the Bethe lattice with an arbitrary coordination number,  using the methods of Lemma \ref{lemma: Pz left-right}, we can show that projecting out a site results in a reduced connectivity for the projector. Consequently, analogous to Lemma \ref{thm:Wannier_for_1d} and Procedure \ref{algo:Gram Schmidt procedure}, using an arbitrary sequence of site locations guarantees the creation of an orthonormal basis which is compactly supported.

\section{Higher dimensional lattices} \label{sec:higher_dimensions}

In 1d, we were able to show that for an arbitrary basis of wavefunctions, the existence of a compactly supported orthogonal basis (CSOB) equals the strict locality of the associated projector. However, the results for the 1d case do not all carry over to higher dimensional lattices. While strict locality of the projector is a necessary condition for the existence of a CSOB even in higher dimensions (see section~\ref{subsec:necessary}), the methods we have employed so far for the 1d case do not prove that it is a sufficient one. Our proof for the existence of a CSOB given any 1d SL projector relied on the fact that any 1d SL projector can be represented as a nearest-neighbor (NN) projector in a supercell representation, or equivalently, a block tridiagonal matrix in the orbital basis. In higher dimensions, such a simple matrix representation for even the simplest non-trivial SL projector, i.e. an NN projector is lacking. Generic SL or NN projectors in $d>1$ cannot be represented by block tridiagonal matrices, or even band diagonal matrices. This makes the task of  identifying the properties of a projector that are equivalent to the existence of a CSOB difficult.

Consequently, we identify a condition more stringent than strict locality of the projector as a sufficient condition for the existence of a CSOB. While NN projectors in $d>1$ cannot in general be represented as block tridiagonal matrices, it is still possible to show that the image of any NN projector is spanned by a CSOB. However, unlike in 1d, higher dimensional SL projectors cannot in general be expressed as NN projectors using a supercell transformation. Consequently, we can extend the results for NN projectors only to those SL projectors that can be brought to an NN form using a supercell representation. We call such projectors NN-reducible projectors. (For a discussion of the condition of being NN-reducible, we refer the reader to section \ref{subsec:NN_reducible}.) Due to these considerations, unlike in 1d, we obtain separate necessary and sufficient conditions for the existence of a CSOB spanning a subspace. Since the necessary condition has already been proved in section~\ref{subsec:necessary}, in this section, we focus on proving the sufficient condition, which is that a projector should be NN hopping, or NN-reducible.

Additionally, as mentioned in the introduction, in higher dimensional lattices it is possible to construct `hybrid Wannier functions' that are compactly supported along just one dimension. We show that their existence hinges on the the strict locality of projectors just like the existence of compactly supported wavefunctions in one dimension.

For all these cases, we provide algorithms for obtaining CSOBs and CS hybrid Wannier functions, and provide upper bounds on their sizes. In summary, we will prove the following statements.

\begin{theorem} \label{thm:thm2}
 Let $P$ be a strictly local projector with a maximum hopping distance $b$, acting on a $d>1$ dimensional tight binding lattice.
 \begin{enumerate}[(1)]
 \item If $P$ is NN reducible, there exists a compactly supported orthonormal basis spanning its image, with each basis vector having a size of at most $3b\times \dotsc \times 3b$ cells.
 \item If $P$ is translationally invariant and NN reducible, there exists a compactly supported Wannier basis spanning its image, in a size $2b\times \dotsc \times 2b$ supercell lattice representation.
 \item If $P$ is translationally invariant, its image is spanned by hybrid Wannier functions which have compact support along the localized (i.e. Wannier-like) dimension, in a size $2b$ supercell representation of the lattice. The supercell transformation is required only along the localized dimension, which may be chosen to be any of the $d$ dimensions.
 Moreover, if $P$ is strictly local along any one direction with a maximum hopping distance $b$, (with no restrictions on the localization along the other directions), then strictly local hybrid Wannier functions corresponding to a size $2b$ supercell which are localized along that direction can be formed.
 \end{enumerate}
 \end{theorem}
In addition, we will discuss the topological properties of such projectors in this section. 
 
Before proving these statements, we first discuss some model Hamiltonians with CS Wannier functions. In figure \ref{fig:2d_fb_lattices}, we show two examples from the literature of 2d flat-band Hamiltonians, that have flat bands spanned by orthogonal CLSs, i.e. flat-band CS Wannier functions. These Hamiltonians are based on the square kagome lattice~\cite{Siddharthan2001, Kuno2020} with six sites per unit cell and the frustrated bilayer~\cite{Richter2006, Derzhko2015} with two sites per unit cell. In both cases the natural choices of the primitive cells are such that each CLS lies entirely within a unit cell. While CS Wannier functions can be easily constructed for the flat bands in these examples, for many flat-band Hamiltonians, orthogonal CLSs spanning a flat band do not exist. However, as discussed in section \ref{subsec:construction_nonortho}, it is possible to modify such Hamiltonians by enlargement of the unit cell, so that orthogonal CLSs span an entire flat band. This approach was used for example in reference~\cite{Kuno2020}, to construct two Hamiltonians on the kagome lattice wherein a flat band is spanned by CS Wannier functions. One can apply the same technique to other flat-band models from the literature, such as the decorated square lattice~\cite{tasaki1992ferromagnetism} and the dice lattice~\cite{Sutherland1986}.
 
In all of these examples, the CS Wannier functions are associated with flat bands, and are localized within one cell each. Consequently, the corresponding band projector is on-site hopping, and hence is independent of $k$ in a k-space representation. One can simply diagonalize the projector in order to obtain CS Wannier functions in such cases. However, there are many Hamiltonians with a band spanned by CS Wannier functions that are spread across multiple cells, and with band projectors that are strictly local, but not on-site hopping. While one can diagonalize a translationally invariant NN projector expressed in k-space, in general a Fourier transform of the obtained Bloch wavefunction results in exponentially localized Wannier functions as opposed to CS Wannier functions (for an example, see Appendix \ref{subsec:1d_example}). Similarly, as discussed in section \ref{subsec:construction_nonortho}, it may be possible to construct CS Wannier functions using destructive interference arguments for translationally invariant NN projectors if they have a simple form. However, it is not straightforward to do so for a more complicated projector, such the following 2d NN projector:
\begin{align}
P(k) = \frac{1}{12}\left(
\begin{smallmatrix}
 e^{-i k_y} \sqrt{3}+e^{i k_y} \sqrt{3}+6 & -e^{-i k_y} \sqrt{3}+e^{i k_y} \sqrt{3}+2 e^{i k_x}-2 & e^{-i k_y} \sqrt{3}+2 e^{i k_x}+1 & -e^{i k_y} \sqrt{3}+2 e^{i k_x}+1 \\
 e^{-i k_y} \sqrt{3}-e^{i k_y} \sqrt{3}+2 e^{-i k_x}-2 & -e^{-i k_y} \sqrt{3}-e^{i k_y} \sqrt{3}+6 & e^{-i k_y} \sqrt{3}-2 e^{-i k_x}-1 & e^{i k_y} \sqrt{3}+2 e^{-i k_x}+1 \\
 e^{i k_y} \sqrt{3}+2 e^{-i k_x}+1 & e^{i k_y} \sqrt{3}-2 e^{i k_x}-1 & 7-2 e^{-i k_x}-2 e^{i k_x} & e^{-i k_x} \left(-e^{i (k_x+k_y)} \sqrt{3}-2 e^{2 i k_x}+2\right) \\
 -e^{-i k_y} \sqrt{3}+2 e^{-i k_x}+1 & e^{-i k_y} \sqrt{3}+2 e^{i k_x}+1 & -e^{-i k_y} \sqrt{3}-2 e^{-i k_x}+2 e^{i k_x} & 5+2 e^{-i k_x}+2 e^{i k_x} \\
\end{smallmatrix}\right). \label{eq:2d_NN_P_example}
\end{align}
In contrast, the method presented in this section enables us to construct CS Wannier functions since this projector is NN hopping. The Wannier functions (centered at ($x,y$)) so obtained are $\ket{W_1} = \frac{1}{\sqrt{6}} \ket{x,y} \otimes (\ket B + \ket C + \ket D) + \frac{1}{\sqrt{6}} \ket{x+1,y} \otimes (\ket A - \ket C +  \ket D)$ and $\ket{W_2} = \frac{1}{2} \ket{x,y} \otimes (\ket A  - \ket B + \ket C) + \frac{1}{2\sqrt{3}} \ket{x,y-1} \otimes (\ket A + \ket B - \ket D)$, where the orbitals are labelled by letters A to D.

\begin{figure}[t]
    \centering
    \begin{subfigure}[t]{0.25\textwidth}
        \centering
        \includegraphics[width=\textwidth]{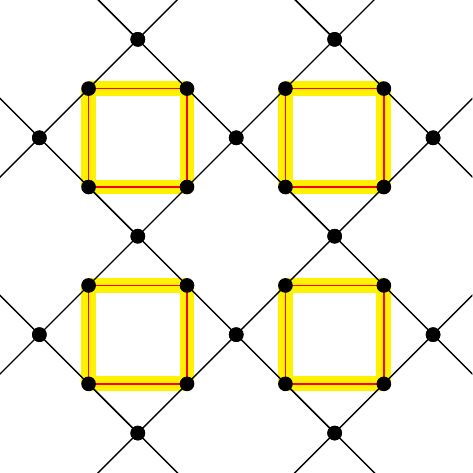}
        \caption{}
        \label{fig:Square_kagome_lattice}
    \end{subfigure}%
     \vspace{10pt}
    \begin{subfigure}[t]{0.3\textwidth}
        \centering
        \includegraphics[width=\textwidth]{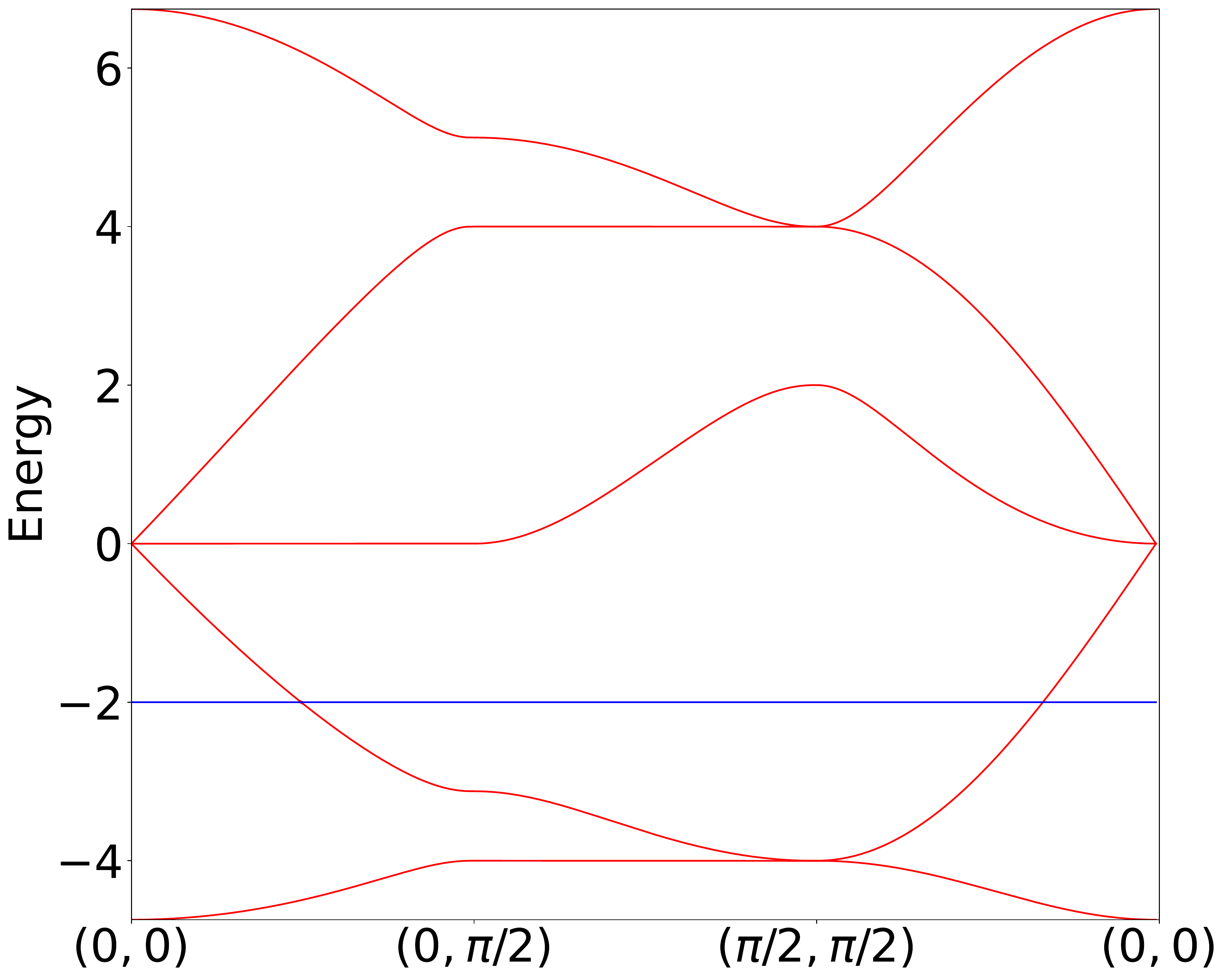}
        \caption{}
        \label{fig:square_kagome_bands}
    \end{subfigure}
  	
  	\begin{subfigure}[t]{0.3\textwidth}
        \centering
        \includegraphics[width=\textwidth]{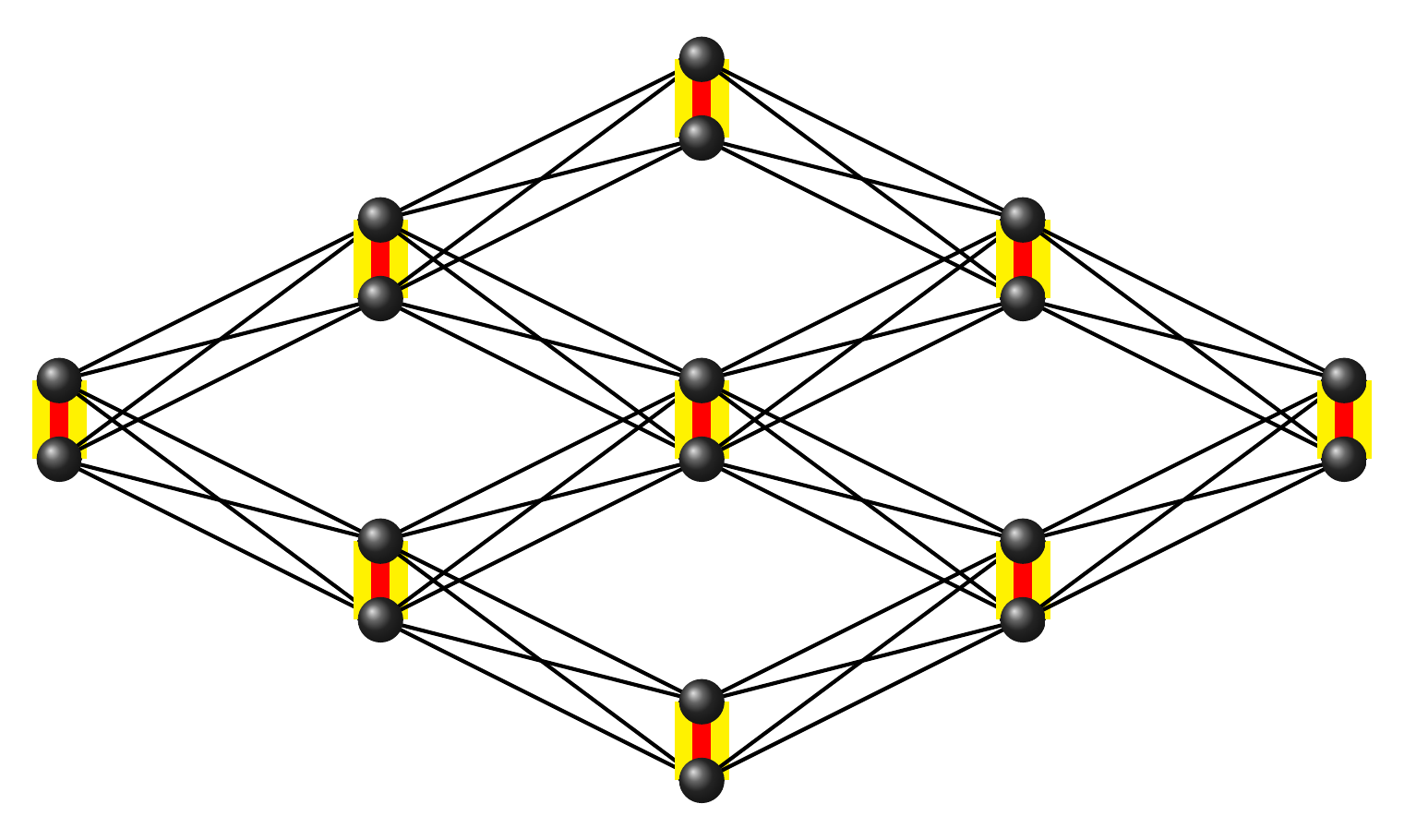}
        \caption{}
        \label{fig:Frust_bilayer}
    \end{subfigure}
 	\vspace{10pt}
    \begin{subfigure}[t]{0.3\textwidth}
        \centering
        \includegraphics[width=\textwidth]{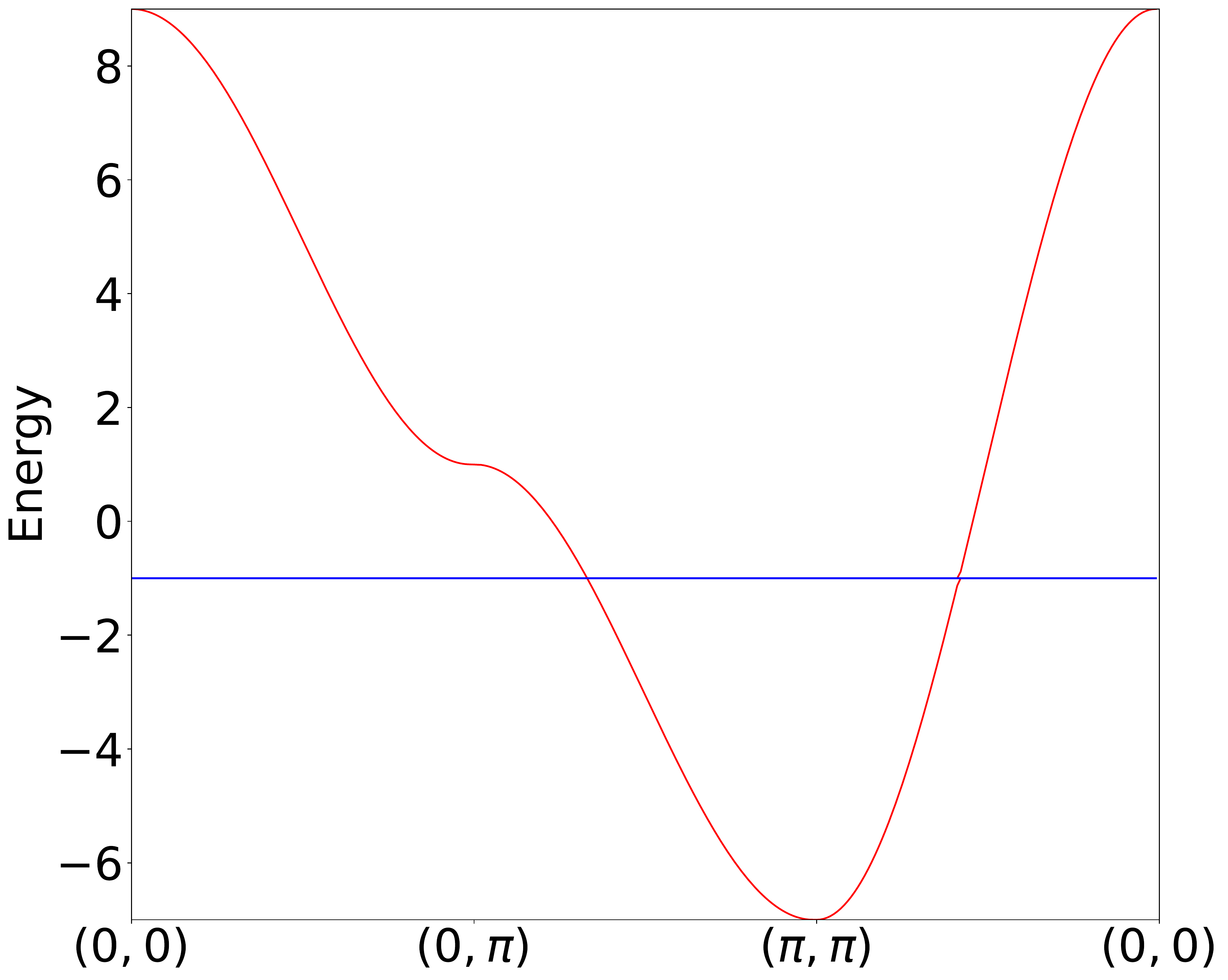}
        \caption{}
        \label{fig:frust_bilayer_line}
    \end{subfigure}
    \caption{(a) The square kagome lattice and (c) the frustrated bilayer lattice. In both cases, the black and red segments denote hopping elements with values $t_1$ and $t_2$ respectively. The flat band CLSs, i.e. CS Wannier functions are highlighted in yellow. In the case of the square kagome lattice, each CLS has support on four sites and an amplitude of $\frac{1}{2}$ with alternating signs on the four sites. For the frustrated bilayer lattice, each CLSs has an amplitude of $+1$ and $-1$ on the two sites where it is located. The band structures of the square kagome lattice with $(t_1,t_2)=(2,1)$, and of the frustrated bilayer with $(t_1,t_2) = (0,1)$ are shown in figures (b) and (d). The flat band is colored blue.}
    \label{fig:2d_fb_lattices}
\end{figure}

Having discussed some examples, we proceed to the proof of theorem \ref{thm:thm2}, which we split across the following subsections. Since the problem of finding hybrid Wannier functions can be reduced to a one dimensional problem, we start with the proof of point (3) of Theorem \ref{thm:thm2} in subsection \ref{subsec:hybrid_Wannier}. In subsections \ref{subsec:NN higher d} and \ref{subsec:NN higher d trans inv} we present procedures for obtaining CSOBs for NN projectors. In subsection \ref{subsec:NN_reducible}, we discuss how to determine whether a projector is NN-reducible, and show how a supercell representation can be used to extend the results for NN projectors to the more general class of NN-reducible projectors. In subsection \ref{subsec:top_triv}, we discuss the topological properties of SL projectors as well as projectors associated with CSOBs and CS hybrid Wannier functions.

\subsection{Hybrid Wannier Functions for Strictly Local Projectors} \label{subsec:hybrid_Wannier}

\begin{algorithm}[h]
    \caption{Compactly Supported Hybrid Wannier functions for Strictly Local Projectors} \label{algo:hybrid_wannier}
    \SetKwInOut{Input}{Input}
    \SetKwInOut{Output}{Output}
    \Input{A translationally invariant SL projection operator $P$ operating on a $d$ dimensional lattice}
    \emph{Procedure:} To obtain hybrid Wannier functions which are Wannier-like along the $d^{th}$ dimension, for every value of $\vec{k}_\perp$ in the $d-1$ dimensional B.Z.:
    \begin{enumerate}
    \item Obtain the 1d projector $P(\vec{k}_\perp)$ using expression \eqref{eq:P_k_perp}.
    \item Use a size $b$ supercell representation to express $P(\vec{k}_\perp)$ as a 1d NN projector.
    \item Following Procedure \ref{algo:GS_1d_Wannier}, obtain a CW basis for the image of $P(\vec{k}_{\perp})$. Let $\widetilde{\p}_{\perp}^{\vec{k}_\perp}$ denote this basis.
    \item Obtain the set $\widetilde{\p}^{\vec{k}_\perp} \coloneqq \{ \vec{k}_\perp \otimes \ket{\psi} : \ket{\psi} \in \widetilde{\p}_{\perp}^{\vec{k}_\perp}\}$.
    \end{enumerate}
    Obtain the set 
    \begin{align*}
    \widetilde{\p} \coloneqq \bigcup_{\vec{k}_\perp}^{B.Z._{d-1}} \widetilde{\p}^{\vec{k}_\perp}.
    \end{align*}
    
    \Output{The set $\widetilde{\p}$ consisting of hybrid Wannier functions within a size $1\times \dotsc\times 1 \times 2b$ supercell representation, which are compactly supported and Wannier-like along the $d^{th}$ dimension.}    
\end{algorithm}

As discussed in the introduction, hybrid Wannier functions are a variant of Wannier functions for $d>1$ dimensional translationally-invariant systems. Hybrid Wannier functions can be obtained by taking the inverse Fourier transform of the Bloch wavefunctions along exactly one dimension. Such wavefunctions can be chosen to be localized and Wannier-like along one dimension, and Bloch wave-like and delocalized along the other dimensions. Starting with a real space representation of a translationally invariant SL projector, we outline a procedure for obtaining such a basis, so that it is compact localized along the localized dimension. We will use the convention \eqref{eq:Hilbert is tensored} of expressing the Hilbert space as a tensor product,  $ \z^{\otimes d } \otimes \h $,  throughout this section, where $\h$ represents the space of orbitals and spin. 

First, we revisit the k-space, or Fourier space representation for a projector. As is customary, we consider finite periodic lattices of size $L_1\times \dotsc \times L_d$, so that the number of sites $N=L_1\dotsc L_d$. For infinite lattices, we take the limit of the lengths going to infinity. Without loss of generality, we choose the direct lattice to be a hyper-cube so that the Brillouin Zone (B.Z.) consists of reciprocal lattice vectors $\vec{k}$ satisfying $\vec{k}.\vec{R} \in 2\pi \z$ for any $\vec{R} \in \z^d $ such that $k_i \in (-\pi,\pi]\ \forall \ i\in \{ 1,\dotsc,d\}$. An orthogonal band projection operator $P$ can be expressed as 
\begin{align*}
P &= \sum_{\vec{k} \in B.Z.} \ket{\smash{\vec{k}}} \bra{\smash{\vec{k}}} \otimes P(\vec{k}), \\
\text{where } \ket{\smash{\vec{k}}} &\coloneqq \frac{1}{\sqrt{N}} \sum_{\vec{r}\in \z^d} e^{-i\vec{k}.\vec{r}} \ket{\vec{r}}, \\
\text{and } P(\vec{k}) &\coloneqq  \bra{\smash{\vec{k}}} P \ket{\smash{\vec{k}}}.
\end{align*}
 We have suppressed all orbital quantum numbers in the expressions above to aid readability. Each $P(\vec{k})$ is an $n \times n$ matrix function of $\vec{k}$. Since $P$ is idempotent, $P(\vec{k}) P(\vec{k}') = \delta_{\vec{k},\vec{k}'} P(\vec{k})$, i.e. all the $P(\vec{k})$'s for distinct $\vec{k}$'s are mutually orthogonal projection operators.
 
In order to obtain hybrid Wannier functions which are localized along the $m^{th}$ dimension, we express $P$ in the Fourier space corresponding to all spatial dimensions, except for the $m^{th}$ dimension. Here, we only discuss the case with $m=d$; the rest can be obtained by simple modifications. Let $B.Z._{d-1}$ denote the Brillouin zone in $d-1$ dimensions. Denoting the spatial position along the $d^{th}$ dimension by $z$, we obtain
\begin{align}
P &= \sum_{\vec{k}_\perp \in B.Z._{d-1}}  \ket{\smash{\vec{k}_\perp}}\bra{\smash{\vec{k}_\perp}} \otimes P(\vec{k}_\perp), \label{eq:P split p perp}
\end{align}
with $\vec{k}_\perp \equiv (k_1,\dotsc,k_{d-1})$, and $P(\vec{k}_\perp)$ being an orthogonal projection operator given by 
\begin{align}
P(\vec{k}_\perp) \coloneqq \sum_{z,z' \in \z}  \ket{z}\bra{z'} \otimes \bra{\smash{\vec{k}_\perp},z} P \ket{\smash{\vec{k}_\perp},z'}.  \label{eq:P_k_perp}
\end{align}
Since $P(\vec{k}_\perp) P(\vec{k}_\perp ') = \delta_{\vec{k}_\perp, \vec{k}_\perp ' } P(\vec{k}_\perp)$ , equation \eqref{eq:P split p perp} implies that the task of obtaining an orthogonal basis for the image of $P$ can be split into the task of obtaining orthogonal bases for each $P(\vec{k}_\perp)$ individually. Since $P$ is strictly local, each $P(\vec{k}_\perp)$ can be thought of as being a one dimensional strictly local projector acting on a lattice with positions $z\in \z$. Corollary \ref{corollary:1d_Wannier} guarantees that each $P(\vec{k}_\perp)$ must have a compactly supported Wannier basis in a size $2b$ supercell representation. This leads us to a procedure of obtaining hybrid Wannier functions which are compact localized along any chosen dimension (see procedure \ref{algo:hybrid_wannier}).

While in our considerations so far, we have considered an SL projector, if we were to consider a projector that was strictly local along any one direction without the requirement that it be strictly local along any of the other directions, it follows from the arguments above that hybrid Wannier functions which are localized in one direction can still be constructed. 
This completes our proof for part one of Theorem \ref{thm:thm2}.

\subsection{Nearest Neighbor Projectors} \label{subsec:NN higher d}
We now construct a CSOB for nearest-neighbour projectors in $d$ dimensions. Similar to the 1d case, the Gram-Schmidt orthogonalization must be carried out in a sequence which guarantees that after any step, the reduced projector remains NN hopping or on-site hopping, if the original projector is NN hopping. We observe that unlike in the 1d case, this puts restrictions on the orthogonalization sequence in higher dimensions. Although there are multiple possible types of sequences which ensure this condition is satisfied, for concreteness, here, we present a specific one (in Procedure \ref{algo:GS_higher_d}), which also be used to obtain a CW basis for translationally invariant NN projectors as well. \par 

This procedure relies on dividing the lattice into two sets, $A$ and $B$ consisting of alternating cells similar to what was done in section \ref{subsec:Wannier_1d}. For concreteness, we choose $A$ and $B$ to be given by
\begin{align}
\begin{split}
A &= \{ \vec{r} : \vec{r} \in \z^d; \sum_{i=1}^d r_i \in 2\z\}; \\
B &= \{ \vec{r} : \vec{r} \in \z^d; \sum_{i=1}^d r_i \in 2\z + 1\} .
\end{split}   \label{eq:A and B higher dim}
\end{align}
Since any two distinct cells $\vec{r}_1,\vec{r}_2 \in A$ are separated by at least two hops, $P_{\vec{r}_1} P_{\vec{r}_2} = 0$, and hence $P_A \coloneqq \sum _{\vec{r}\in A} P_{\vec{r}}$ is an orthogonal projector.

\begin{lemma} \label{lemma: PA higher dims}
The orthogonal projector $P-P_A$ satisfies $\bra{\vec{r}_1,i} (P-P_A) \ket{\vec{r}_2,j}=0$ for all $i,j\in \{1,\dotsc,n\}$, unless $\vec{r}_1 = \vec{r}_2 \in B$.
\end{lemma}
\begin{proof}
Let $\hat{\delta}_i$ denote the unit vector along dimension $i$. Similar to the proof for 1d projectors (Lemma \ref{lemma: Pz left-right}), we introduce the diagonal basis (cf \eqref{eq:diagonal_basis_1d}), with the primes dropped for notational convenience. As before, we denote the diagonal of the $P_{\vec{r}\vec{r}}$ matrix in this representation by $d_{\vec{r}}$, so that $\bra{\vec{r},\alpha}P\ket{\vec{r},\beta} = (d_{\vec{r}})_\beta \delta_{\alpha \beta}$. \par 

Since $P\ket{\vec{r},\alpha} = P_A \ket{\vec{r},\alpha}$ whenever $\vec{r}\in A$, both $\vec{r}_1$ and $\vec{r}_2$ must belong to $B$ for the corresponding matrix element to be non-zero. If $\vec{r}_1, \vec{r}_2 \in B$ and are distinct, we get 
\begin{align}
\bra{\vec{r}_1,\alpha} P- P_A \ket{\vec{r}_2,\beta} &= - \bra{\vec{r}_1,\alpha} P_A \ket{\vec{r}_2,\beta}  \nonumber \\
&=  - \bra{\vec{r}_1,\alpha} \sum_{s = \pm 1}\sum_{m=1}^d P_{\vec{r}_2+ s \hat{\delta}_m}  \ket{\vec{r}_2,\beta}, \label{eq:higher_d_expr_interm}
\end{align}
since $P_{\vec{r}} \ket{\vec{r}_2,\beta} =0$ unless $\vec{r}$ is a nearest neighbor of $\vec{r}_2$. Since $P$ only has NN hopping terms, this is zero, unless $\vec{r}_1$ and $\vec{r}_2$ are equal to each other, or are two hops away from each other, i.e. only if $\vec{r}_2$ is of the form $\vec{r}_1 \pm \hat{\delta}_p \pm \hat{\delta}_q$ for some $p, q\in \{1,\dotsc,d\}$. In order to show that the matrix element is zero for the case with two hops, we will require the higher dimensional analogs of equations \eqref{eq:1d_self_hop_positive_def}, \eqref{eq:1D_sum_of_hoppings_is_1} and \eqref{eq:1D_sum_prod_is_zero}, which are
\begin{align}
\bra{\vec{r},\alpha} P \ket{\vec{r},\alpha} &\geq 0 ,\\
 \bra{\vec{r},\alpha} P \ket{\vec{r},\alpha}  + \bra{\smash{ \vec{r} + \hat{\delta}_p,\beta} } P \ket{ \smash{   \vec{r} + \hat{\delta}_p,\beta   }    } &= 1 \ \forall p\in \{ 1,\dotsc,n\}, \text{ whenever } \bra{\vec{r},\alpha} P \ket{ \smash{   \vec{r} + \hat{\delta}_p,\beta } }   \neq 0,  \label{eq:higher_D_sum_hops_1}\\
 \text{and }\sum_\gamma \sum_{\vec{v}}^{\text{c.n.}}  \bra{\vec{r},\alpha} P \ket{\vec{v},\gamma}  \bra{\vec{v},\gamma}  P \ket{\vec{w},\beta} &= 0 \label{eq:higher_D_sum_prod_is_zero}
\end{align}
respectively. The summation in \eqref{eq:higher_D_sum_prod_is_zero}, with a superscript `c.n.' (for common neighbors) is over those vectors $\vec{v}$ which are nearest neighbors of both $\vec{r}$ and $\vec{w}$.  \par 

For the case where $\vec{r}_2$ is two hops away from $\vec{r}_1$, expression \eqref{eq:higher_d_expr_interm} simplifies to zero as follows:
\begin{align*}
\bra{\vec{r}_1,\alpha} \sum_{m=1}^d P_{\vec{r}_2+\hat{\delta}_m}  \ket{\vec{r}_2,\beta} &= \bra{\vec{r}_1,\alpha} \sum_{\vec{w}}^{\text{c.n.}} P_{\vec{w}}  \ket{\vec{r}_2,\beta} \\
&=  \sum_{\vec{w}}^{\text{c.n.}} \sum_{\gamma}^{(d_{\vec{w}})_{\gamma} \neq 0} \frac{\bra{\vec{r}_1,\alpha} P \ket{\vec{w},\gamma}    \bra{\vec{w},\gamma}   P   \ket{\vec{r}_2,\beta}}{\bra{\vec{w},\vec{\gamma}}P \ket{\vec{w},\gamma}} .
\end{align*}
If $\bra{\vec{r}_1, \alpha}P \ket{\vec{r}_1,\alpha} = 1$, then every term in the summation is zero. Otherwise, we get
\begin{align*}
\dotsc  &=  \frac{(\delta_{\vec{r}_1})_{\alpha}}{1-\bra{\vec{r}_1, \alpha}P \ket{\vec{r}_1,\alpha}} \sum_{\vec{w}}^{\text{c.n.}} \sum_{\gamma} \bra{\vec{r}_1,\alpha} P \ket{\vec{w},\gamma}    \bra{\vec{w},\gamma}   P   \ket{\vec{r}_2,\beta} \tag*{from \eqref{eq:higher_D_sum_hops_1},} \\
&= 0    \tag*{from \eqref{eq:higher_D_sum_prod_is_zero}.},
\end{align*}

Thus, matrix elements of $P-P_A$ can be non-vanishing only if $\vec{r}_1=\vec{r}_2 \in B$.
\end{proof}

\begin{algorithm}[t]
    \caption{Construction of a CSOB for any NN projector in arbitrary dimensions} \label{algo:GS_higher_d}
    \SetKwInOut{Input}{Input}
    \SetKwInOut{Output}{Output}
    \Input{A nearest neighbor projector $P$ acting on a $d \geq 1$ dimensional lattice}
    \emph{Procedure:} Divide the lattice into two sets $A$ and $B$ consisting of alternating cells, according to \eqref{eq:A and B higher dim}.
    \begin{enumerate}
    	\item Obtain $\widetilde{\p}_{\vec{r}}^P$ and hence $P_{\vec{r}}$, for every $\vec{r}\in A$.
    	\item Obtain the orthogonal projector $P - P_A \coloneqq P - \sum_{\vec{r}\in A} P_{\vec{r}}$.
    	\item Obtain $\widetilde{\p}_{\vec{r}}^{P-P_A}$ for every $\vec{r}\in B$.
    \end{enumerate}
    
    \Output{The set $\tilde{\Pi} \coloneqq  \left(   \bigcup_{\vec{r} \in A} \widetilde{\p}_{\vec{r}}^P \right)   \cup  \left( \bigcup_{\vec{r} \in B} \widetilde{\p}_{\vec{r}}^{P-P_A} \right)$ which is a CSOB spanning the image of $P$.}    
\end{algorithm}

The basis vectors obtained using Procedure \ref{algo:GS_higher_d} consist of wavefunctions which have a maximum spatial extent (volume) of at most $3\times\dotsc \times 3$ cells. As shown in figure \ref{fig:2D procedure}, they are in fact significantly smaller in extent than this upper bound.

\begin{figure*}[]
    \centering
    \begin{subfigure}[t]{0.45\textwidth}
        \centering
        \includegraphics[scale=1]{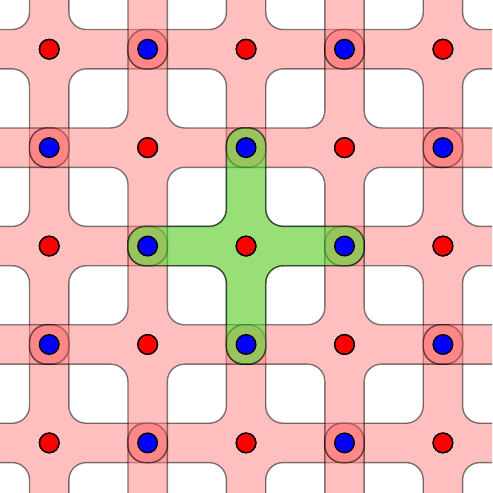}
        \caption{}
        \label{fig:2DStep1}
    \end{subfigure}%
    \begin{subfigure}[t]{0.45\textwidth}
        \centering
        \includegraphics[]{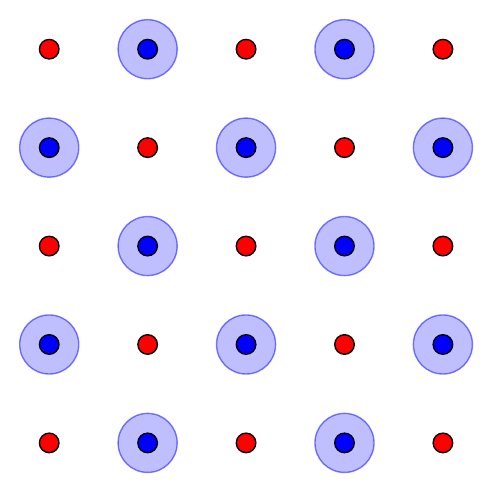}
        \caption{}
        \label{fig:2DStep2}
    \end{subfigure}
    \caption{{\it Procedure \ref{algo:GS_higher_d} for a 2d NN projector}: Cells belonging to sets A and B are represented by red and blue dots respectively. (a) In order to obtain $P_A$, we operate $P$ on each cell belonging to $A$, and orthogonalize the vectors. Each colored shape centered at location $\vec{r}$ denotes the maximum spatial extent of the wavefunctions in $\widetilde{\p}_{\vec{r}}^P$.  (We show one shape in green, in order to highlight the `plus' shape of each of these sets). (b) Since $P-P_A$ is on-site hopping, operating it on any cell in set $B$ creates wavefunctions which are localized at exactly that cell. Each blue circle centered at $\vec{r} \in B$ represents the maximum extent of vectors in $\widetilde{\p}_{\vec{r}}^{P-P_A}$.}
    \label{fig:2D procedure}
\end{figure*}

\subsection{Translationally Invariant Nearest Neighbor Projectors} \label{subsec:NN higher d trans inv}
Let $\y_i$ denote the unit translation operator along the $i^{th}$ dimension. For $\vec{r} \in \z^d$, let $\y_{\vec{r}}$ denote a translation by an amount $\vec{r}$. \par 
Although Procedure \ref{algo:GS_higher_d} can also be used to generate a compactly supported orthonormal basis from translationally invariant projectors, in general the resulting basis will not be a Wannier basis. However, the wavefunctions can be chosen to have translational invariance properties if each $\widetilde{\p}_{\vec{r}}^P$ ($\widetilde{\p}_{\vec{r}}^{P-P_A}$) for $\vec{r}\in A$ ($\vec{r}\in B$) is chosen to be a translation of $\widetilde{\p}_{\vec{0}}^P$ ($\widetilde{\p}_{\hat{\delta}_1}^{P-P_A}$) by $\vec{r}$ ($\vec{r}-\hat{\delta}_1$) cells. Based on this principle, we propose a method (Procedure \ref{algo:GS_higher_d_Wannier}) for obtaining a CS Wannier basis spanning the image of $P$ within a supercell representation. For $d=1$, this procedure is equivalent to Procedure \ref{algo:GS_1d_Wannier}.

\begin{algorithm}[H]
    \caption{Compactly Supported Wannier Basis for $d$ dimensional NN Projectors} \label{algo:GS_higher_d_Wannier}
    \SetKwInOut{Input}{Input}
    \SetKwInOut{Output}{Output}
    \Input{A translationally invariant NN projection operator $P$ on a $d$ dimensional lattice}
    \emph{Procedure:} 
    \begin{enumerate}
    	\item Obtain $\p_0^P$, and orthogonalize it to obtain the set $\widetilde{\p}_0^P$.
    	\item Obtain the orthogonal projection operator $P_0$ onto the span of $\p_0^P$.
    	\item Obtain a reduced projection operator (which removes all the nearest neighbors of the site $\hat{\delta_1}$)
    	\begin{align*}
    	    P'\coloneqq P - \sum_{i=1}^d\y_i \y_1 P_0 \y_1^{\dagger} \y_i ^{\dagger} - \sum_{i=1}^d\y_i^{\dagger} \y_1 P_0 \y_1^{\dagger} \y_i.
    	\end{align*}
    	\item Obtain and orthogonalize $\p_{\hat{\delta}_1}^{P'}$ to obtain the set $\widetilde{\p}_{\hat{\delta}_1}^{P'}$.
    	\item Obtain the set $\p$, defined as 
    	\begin{align}
    		\p = \left( \bigcup_{ \vec{r} \in A } \{ \y_{\vec{r}} \ket{\chi}:\ket{\chi}\in \widetilde{\p}_0^P \} \right) \cup \left( \bigcup_{ \vec{r} \in B} \{ \y_{\vec{r} - \hat{\delta}_1} \ket{\chi}:\ket{\chi}\in \widetilde{\p}_{\hat{\delta}_1}^{P'} \}  \right) \label{eq:wannier_higher_d_union}
    	\end{align}
    \end{enumerate}
    \Output{The set $\Pi$ consisting of compactly supported Wannier functions spanning the image of $P$, within a size $2\times \dotsc \times 2 $ supercell representation.}
\end{algorithm}    

Using steps similar to those used for proving theorem \ref{thm:Wannier_for_1d}, we can show that Procedure \ref{algo:GS_higher_d_Wannier} outputs a CW basis as claimed.

\subsection{Supercell representation and Nearest Neighbor Reducible Projectors} \label{subsec:NN_reducible}
We now discuss how to extend the results and methods for NN projectors to SL projectors with larger hopping distances. To that end, we use a supercell representation analogous to the one used for 1d in section \ref{subsec:conversion}. Unlike in 1d where every SL projector is an NN projector in some supercell representation, in general, an SL projector in $d>1$ becomes a \emph{next}-nearest-neighbor (NNN) hopping projector instead of an NN projector using this transformation. Specifically, if the maximum hopping distance of an SL projector is $b$, then the projector has at most NNN hopping terms within a size $b\times \dotsc \times b$ supercell representation (see figure \ref{fig:2DRegrouping_bigfig}). This reversible transformation is associated with the correspondence:
\begin{align}
\begin{split}
\text{Primitive cell representation} &\longleftrightarrow \text{Supercell representation} \\
\ket{r_1,\dotsc,r_d,i} & \ \equiv \ \  \Big |  r_1 \setminus b, \dotsc, r_d \setminus b,  i + \sum_{m=1}^d n^m (r_m \bmod b)  \Big \rangle  _s,  
\end{split} \label{eq:Corresp_super_primitive_higher_d}
\end{align}

\begin{figure}[t]
    \centering
    \begin{subfigure}[t]{0.45\textwidth}
        \centering
        \includegraphics[scale=0.75]{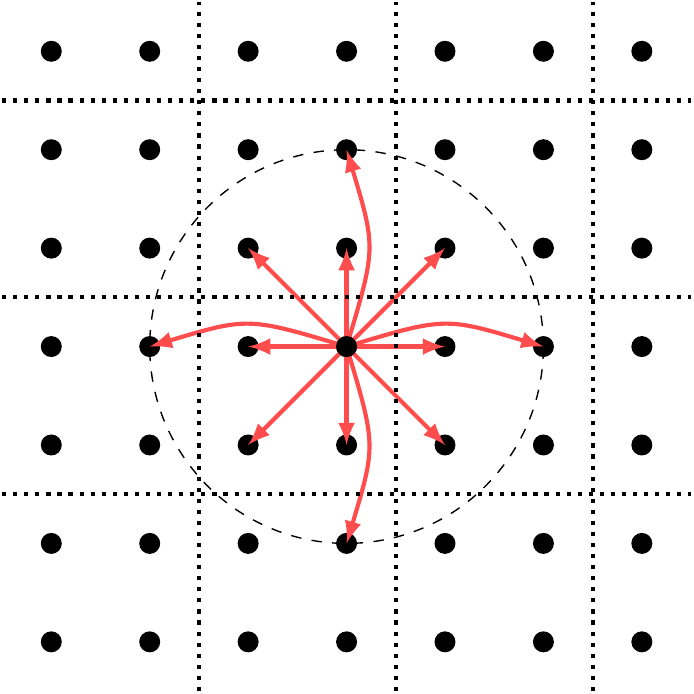}
        \caption{}
        \label{fig:2Dregrouping}
    \end{subfigure}%
    ~
    \begin{subfigure}[t]{0.45\textwidth}
        \centering
        \includegraphics[]{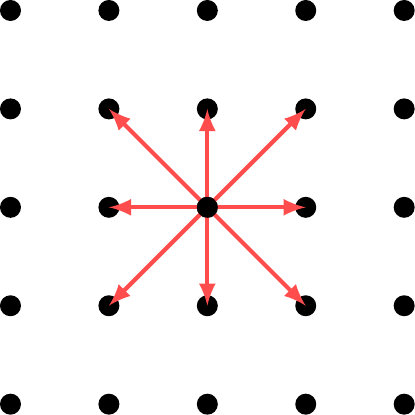}
        \caption{}
        \label{fig:2Dpostregrouping}
    \end{subfigure}
    \caption{Each black dot represents a lattice cell (which may consist of multiple orbitals). The dashed circle denotes the maximum hopping distance from that cell. (a) The connectivity of a generic SL operator with a maximum hopping distance of $b=2$. Hopping elements from an arbitrary cell are shown in red.Grouping all the sites within each cell of a $2\times 2$ grid, we obtain a supercell representation. (b) The operator becomes an NNN hopping operator in the supercell representation. Here, the operator connects neighboring cell along the $\hat{x}$, $\hat{y}$, as well as the $\hat{x} \pm \hat{y}$ directions, with $\hat{x}$ and $\hat{y}$ denoting the two axes.}
    \label{fig:2DRegrouping_bigfig}
\end{figure}

Since the methods developed in sections \ref{subsec:NN higher d} and \ref{subsec:NN higher d trans inv} are applicable only to NN projectors (and not general NNN projectors), even after using a supercell representation, these techniques cannot be applied to a general SL projector operator for $d>1$. 

Hence, we look for projectors that have an NN-hopping form using suitable transformations. We call such projectors NN-reducible. Which SL projectors are NN-reducible? While SL projectors with various types of connectivity are NN-reducible, here we highlight some particularly simple cases. It is easy to see that SL projectors that satisfy $\bra{\vec{r},i}P\ket{\vec{r}\ ',j} = 0 $, whenever $ \vec{r} - \vec{r} \ '$ is not along any of the primitive cell directions become NN hopping within a supercell representation of the type described above. For example, an SL projector on a square lattice that has non-zero hoppings only along the $\pm\hat{x}$ and $\pm\hat{y}$ directions, but no other direction is NN-reducible. This is in contrast with the NNN-reducible projectors of the type shown in figure \ref{fig:2Dpostregrouping}. For some SL projector that appear to be of an NNN form after the transformation \eqref{eq:Corresp_super_primitive_higher_d}, it may still be possible to apply the results from the previous sections to obtain CS Wannier functions. For example, if an SL projector on a square lattice in a supercell representation only connects neighboring cells along the $\hat{x}$ and $\hat{x}+\hat{y}$ directions, the projector is NN-hopping if we regard these two vectors as being the lattice vectors. We consider all such SL projectors that can be brought to an NN form using change of primitive cell vectors and supercell transformations as NN-reducible. Without loss of generality, we only consider NN-reducible projectors that have hopping elements only along the primitive cell directions below.

For such projectors, we obtain the following results:
\begin{corollary}
If an SL projector with a maximum hopping distance $b$ is NN reducible, its image is spanned by a CSOB consisting of wavefunctions of a maximum spatial extent of $3b\times \dotsc \times 3b$ cells. 
\end{corollary}
The procedure for obtaining such a basis is summarized in the following table:
\begin{center}
\begin{tabular}{ c c c c c c c}
Primitive cell &  $\xrightarrow[b\times\dotsc \times b]{\text{Size}}$ & \multicolumn{3}{c}{Supercell representation} &  $\longrightarrow$ & Primitive cell \\
representation & & & & &  &representation \\
SL Projector  &  $\longrightarrow$ & NN Projector  &  $\xrightarrow{\text{Procedure \ref{algo:GS_higher_d}}}$ & CSOB &  $\longrightarrow$ &  CSOB \\
($b\geq 1$) & & ($b=1$) & & (max size $3$) & & (max size $3b$).
\end{tabular}
\end{center}

\begin{corollary} \label{corollary:nd_Wannier}
If a translationally invariant SL projector with a maximum hopping distance $b$ is NN reducible, its image is spanned by a CS Wannier basis within a size $2b\times \dotsc \times 2b$ supercell representation.
\end{corollary}
The procedure for obtaining such a basis is summarized in the following table:
\begin{center}
\begin{tabular}{c c c c c}
Primitive cell  & $\xrightarrow{\text{Size b}}$ & Size $b$ supercell  & $\xrightarrow{\text{Size 2}}$ & Size $2b$ supercell \\
representation & &  representation & & representation \\
SL Projector ($b\geq 1$) & $\longrightarrow$ &  NN Projector ($b=1$) &  $\xrightarrow{\text{Procedure \ref{algo:GS_higher_d_Wannier}}}$ & CS Wannier basis.
\end{tabular}
\end{center}
This completes our proof for theorem~\ref{thm:thm2}.

\subsection{Topological Triviality and Compact Bases}\label{subsec:top_triv}
So far, we have investigated the conditions under which is it possible for a band to be spanned by compactly supported orthogonal bases (CSOBs) or by compactly supported hybrid Wannier functions. Specifically, we identified localization properties of associated projection operators that are necessary or sufficient for the existence of such bases. Similar questions regarding the existence of localized Wannier functions have a long history, as noted in the introduction. The existence of exponentially localized Wannier functions has a bearing on the topological properties of the corresponding bands, with non-trivial band topology restricting the degree of localization possible for Wannier functions. This motivates us to investigate the topological properties of bands possessing CSOBs.

We start with a brief overview of known results. Thouless~\cite{thouless1984wannier} showed that well localized magnetic Wannier functions can be constructed in 2d if and only if the Chern number is zero. For $d\leq 3$, it was later shown that in the presence of time-reversal symmetry and translational invariance, exponentially localized Wannier functions corresponding to an isolated band or a set of isolated bands always exist~\cite{Kohn1959, dexCloizeauxexpWan, nenciu1983existence, Monaco2015}. For systems with other symmetry properties, the existence of such localized Wannier functions is not guaranteed, or even impossible, when the bands are topologically non-trvial. For instance, bands with non-zero Chern numbers do not possess exponentially localized Wannier functions~\cite{brouder2007exponential}. More generally, there exists a {\it localization dichotomy}~\cite{Monaco2018} which says that either exponentially localized Wannier functions exist and the Chern numbers are zero, or all Wannier functions are delocalized (with diverging second moment of the position operator) and the Chern numbers are non-zero. Recently, this result was partially extended to disordered systems~\cite{marcelli2020localization}, with a proof for the vanishing of the Chern marker for 2d insulators possessing exponentially localized generalized Wannier functions.

Since compactly supported Wannier functions (in $R^n$) are even more localized than generic exponentially local Wannier functions, one may expect topological triviality to follow immediately from these results. While we find this to be true (as discussed below), care is needed while drawing such a conclusion. In all the works discussed above, the wavefunction localization is described in terms of decay of wavefunction amplitude in real space, i.e. $R^n$. We note that this notion of localization may not in general be the same as localization in tight-binding models which we consider in this paper. Specifically, orbitals on the lattice ($\mathbb{Z}^n$) which are used as the basis in tightbinding descriptions may themselves not be compactly supported, or even exponentially localized in space ($R^n$). Thus the wavefunctions which are linear combinations of a finite set of such tightbinding orbitals do not in general vanish outside a certain bounded region in space as one might otherwise assume from the use of the term ``compact support" in describing these wavefunctions.

A number of results relating compact support localization of Wannier functions in tight-binding models and topological triviality of associated bands are relevant to the cases considered in this paper. Specifically, it was proved that flat bands in 2d flat-bands Hamiltonians which are strictly local in a tight-binding sense always have a Chern number of zero~\cite{chen2014impossibility}. This property can be viewed as arising due to the fact that the flat bands in such models are spanned by compactly supported Wannier-type functions or CLSs. More recently, this result was generalized to all symmetry classes and arbitrary dimensions greater than one, by proving that the vector bundle associated with band(s) which are spanned by CS Wannier-type functions are topologically trivial~\cite{Read2017}. As discussed in section \ref{subsec:construction_nonortho}, CS Wannier-type functions are in general non-orthogonal, and consequently the orthogonal Wannier functions we consider in this paper are a special type of CS Wannier-type functions. Similarly, the set of SL projectors is a subset of the set of projectors that have their images spanned by CS Wannier-type functions. Thus, it follows directly from the results in~\cite{Read2017} that:
\begin{theorem}
For translationally invariant tight-binding models in $d>1$, a set of bands that is spanned by compactly supported Wannier functions is topologically trivial. More generally, bands associated with strictly local projectors are topologically trivial.
\end{theorem}

It follows that translationally invariant NN and NN-reducible projectors are topologically trivial, since they are SL.

While a topological obstruction exists for constructing localized Wannier functions, no such obstruction exists for localized hybrid Wannier functions. Since hybrid Wannier functions can be treated as 1d Wannier functions (see section \ref{subsec:hybrid_Wannier}), using the arguments in~\cite{marzari1997maximally}, it follows that for any number of dimensions, hybrid Wannier functions that are exponentially localized along the localized axis exist, independent of the topological properties of the associated band(s). For example, as can be seen using the Coulomb gauge, quantum Hall systems admit localized hybrid Wannier-like solutions that are exponentially decaying along one direction, despite the Chern number being non-zero~\cite{yoshioka2013quantum}. Similarly, anomalous quantum Hall systems possess maximally localized hybrid Wannier functions that are exponentially localized~\cite{Qi2011} despite a non-vanishing Chern number. These statements do no preclude the possibility of a topologically non-trivial band being spanned by compactly supported hybrid Wannier function. An interesting consequence of the theorem above, and the equivalence of strict localization of a projector, and the existence of CS hybrid Wannier functions along all axes, we find that such bands are \emph{necessarily} topologically trivial. Specifically:
\begin{theorem}
For a $d>1$ dimensional system, if a set of bands is such that for any of the d orthogonal axes, there exist hybrid Wannier functions (spanning the bands) that are compactly supported along the chosen axis, then the band(s) are topologically trivial.
\end{theorem}

In this section, we have so far only considered systems which are translationally invariant. It would be interesting to study the topological properties of similar systems without translational invariance. For example, an analog of compactly supported hybrid Wannier functions could be a set of wavefunctions not related by lattice translations that are each compactly supported along one direction, but possibly delocalized along the other directions. We leave the question of topological triviality of such cases for future work.

Based on physical ground we anticipate that even non-translationally invariant SL projectors as well as bands associated with CSOBs should be topologically trivial. Using simple arguments applied to recent results from the literature, we will now show that the latter statement is indeed true. Specifically, we consider projectors that have no symmetries except possibly lattice translation symmetry, i.e. class A systems from the Altland-Zirnbauer classification scheme~\cite{Altland1997}. The topological classification of systems across all dimensions is organized in the form of a periodic table~\cite{Kitaev2009, Ryu2010}. As can be seen from the table, odd dimensional class A projectors are always K-theoretically trivial. However, in $2n$ dimensions, they are characterized by the integer valued $n^{th}$ Chern number~\cite{katsura2018noncommutative}.

Using simple arguments, we will now show that the topological Chern invariant for class A projectors that are associated with CSOBs is zero.
\begin{theorem} \label{thm:ch_is_zero}
In all dimensions, if a projector without symmetries (except possibly lattice translation symmetry) is such that its image is spanned by a compactly supported orthogonal basis, then it is Chern trivial.
\end{theorem}
We only need to show that such projectors in even dimensions are Chern trivial. To that end, we use the real space expression from~\cite{katsura2018noncommutative}, for the integer valued Chern number, which for a $2n$ dimensional system is given by
\begin{align}
\text{Ind } P &= -\frac{(2\pi i)^n}{n!} \sum_\sigma (-1)^\sigma \Tr P [\theta_{\sigma_1},P]\dotsc [\theta_{\sigma_n},P], \label{eq:Koma_index}
\end{align}
where $\Tr$ denotes the trace operation, the summation is over all permutations $\sigma$, and $\theta_i$ denotes the projection operator onto the positive half plane along the $i^{th}$ direction, i.e. $\theta_i = \sum_{\alpha}\sum_{{r_i>0}} \ket{\vec{r},\alpha} \bra{\vec{r},\alpha}$.
Our arguments are based on the following properties of the index: (i) the additivity of the index for mutually orthogonal projectors, (ii) the independence of the index from the choice of the origin, or axes, and (iii) the local computability of the index. Physically, since the $2d$ case is of the most interest, we demonstrate our arguments by applying them to the 2d case. These arguments and the conclusion are valid for higher dimensional cases as well.

First, we note that for 2d, the index is the same as the Chern marker expression~\cite{kitaev2006anyons, marcelli2019haldane}, which is the real space analog of the k-space expression for Chern number used for periodic systems. The Chern marker for a projector $P$ is given by
\begin{align}
\Ch P &= -2\pi i \Tr P [[\theta_x, P],[\theta_y, P]].  \label{eq:Chern2D}
\end{align}
As shown explicitly in~\cite{kitaev2006anyons}, the Chern marker is additive, i.e. Chern marker of the sum of two mutually orthogonal projectors is the sum of the Chern markers of the two projectors.

Consider a CSOB of size $R$, consisting of wavefunctions $\ket{\psi_i}$. Let $P$ be the orthogonal projector onto the span of the CSOB. Clearly, $P$ is SL, and is the sum of the orthogonal projection operators $P_i$ projecting onto states $\ket{\psi}_i$'s:
\begin{align}
\begin{split}
P &= \sum_{i=1} P_i; \\
\text{with }P_i &= \ket{\psi_i}\bra{\psi_i},\\
P_i P_j &=\delta_{ij} P_i. 
\end{split} \label{eq:P splitting}
\end{align}
By definition, each wavefunction $\ket{\psi_i}$ has non-zero support only within a circle $B_i$ centered at some location $\vec{c}_i$ of radius $R$. (The $\vec{c}_i$'s are not unique; however, the conclusions that follow do not depend on the choice.) Thus, each projector $P_i$ has non-zero hopping terms only within $B_i$.

From \eqref{eq:Chern2D}, we note that $\Ch P_i \neq 0$ only if $\ket{\psi_i}$ straddles both the axes, i.e. if $|\vec{c}_i| \leq R$, and is zero otherwise. For example, if $B_i$ lies entirely in the right half plane, then the operator $\theta_x$ can be replaced by the identity operator in equation \eqref{eq:Chern2D}, resulting in a zero Chern marker (and similarly for other cases). Using this, and the additivity of the Chern marker, we obtain
\begin{align*}
\Ch P &= \Ch \underbrace{(\sum_{|\vec{c_i}|\leq R} P_i )}_{\tilde{P}} +  \Ch (\sum_{|\vec{c_i}| >  R} P_i )= \Ch \tilde{P}.
\end{align*}
(For a cartoon picture see figure \ref{fig:Chern local computation}.) Since the Chern number is independent of the choice of the origin, we may shift the origin and reevaluate the Chern number without affecting its value. Consider moving the origin by a distance of at least $2R$ in any direction. For example, let the new location of the origin be $(-3R,-3R)$. Since none of the constituent wavefunctions of $\tilde{P}$ straddle both of the new axes, $\Ch \tilde{P}=0$. Consequently,
\begin{align*}
\Ch P =0.
\end{align*}
The Chern marker is additive in all dimensions. Thus, applying the same reasoning to equation \eqref{eq:Koma_index}, it is easy to show that projectors associated with CSOBs are Chern trivial in all dimensions.

\begin{figure}
    \centering
    \begin{subfigure}{0.35\textwidth}
        \centering
        \includegraphics[width=\textwidth]{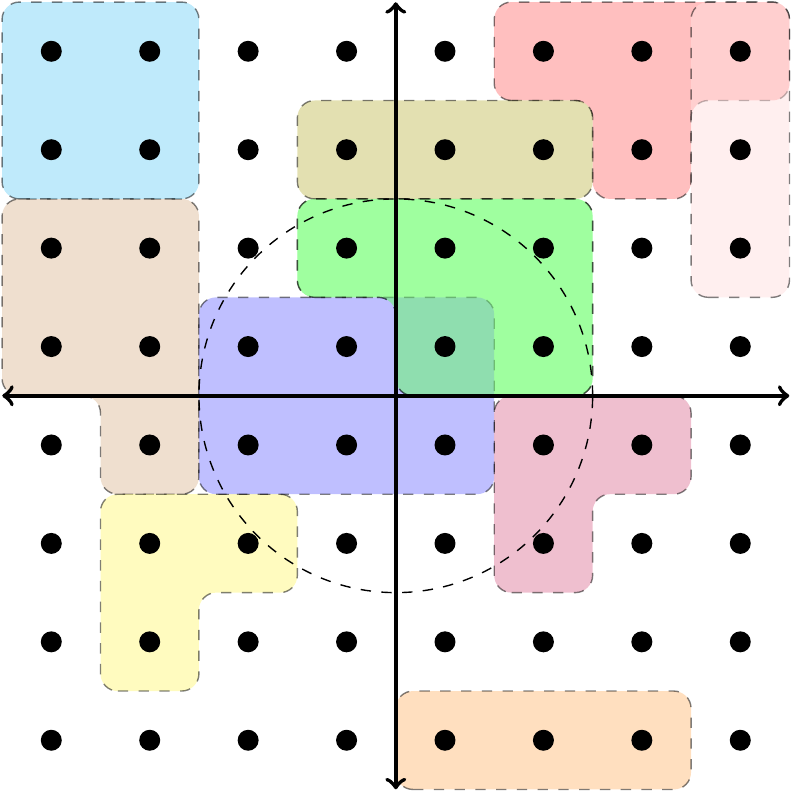}
        \caption{}
        \label{fig:pre box construction}
    \end{subfigure}%
    \hspace{15mm}
    \begin{subfigure}{0.35\textwidth}
        \centering
        \includegraphics[width=\textwidth]{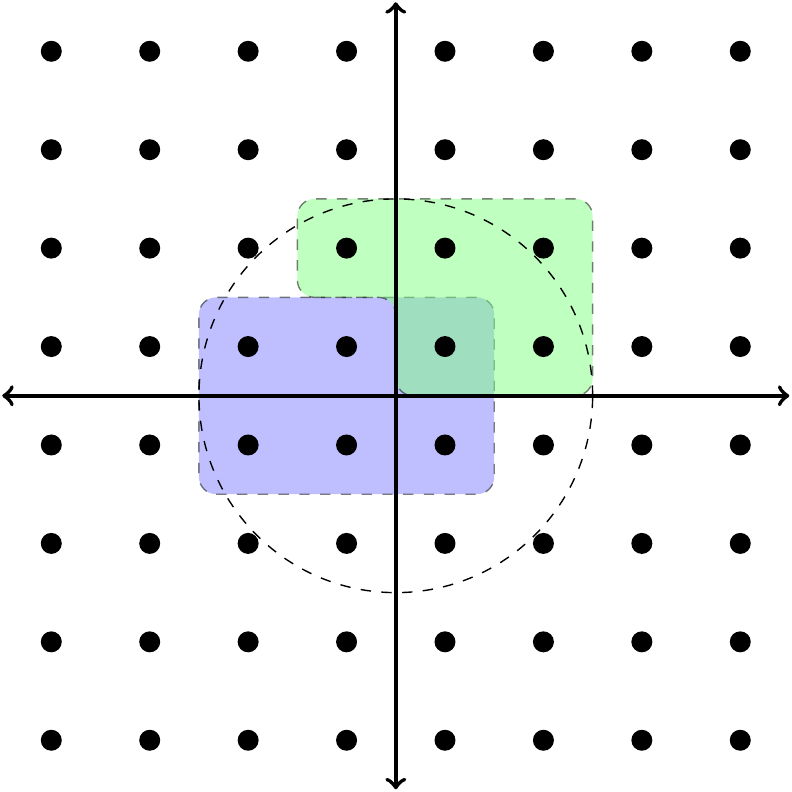}
        \caption{}
        \label{fig:post box construction}
    \end{subfigure}
    \caption{An illustrative example in 2d. (a) Each basis state $\ket{\psi_i}$ of a CSOB of size $R=2$ is shown by a colored region, and has non-zero support only on the sites within that region. (b) Only those wavefunctions that have their centers $\vec{c}_i$'s inside the dashed circle have non-zero contributions to the Chern marker.  Retaining only these wavefunctions defines a new projector $\tilde{P}$, which has the same Chern number as $P$.}
    \label{fig:Chern local computation}
\end{figure}

\section{Conclusions} \label{sec:conclusion}
In this paper, we have obtained necessary and sufficient conditions for a band or a set of bands to be spanned by compactly supported Wannier functions, or in the absence of lattice translational invariance, for a subspace of the Hilbert space to possess an orthogonal basis consisting of compactly supported wavefunctions. In 1d tight-binding models, we have established that there exists a compactly supported orthogonal basis spanning the occupied subspace iff. the corresponding projection operator is strictly local. In the process, we presented an algorithm for constructing a compactly supported orthogonal basis for the image of any strictly local projector. This algorithm generates wavefunctions having a maximum spatial extent three times the maximum hopping distance of the projector. Nearest neighbor projectors on a Bethe lattice with arbitrary coordination number also have a compactly supported orthonormal basis. In the appendix, a simple method for the construction of some strictly local projectors is provided.

For higher dimensional lattices, we showed that while strict locality of the projector is a necessary condition for the existence of a compactly supported orthogonal basis, a sufficient condition is that the projector be expressible as a nearest neighbor projector using a change of primitive cell vectors or a supercell representation. For such projectors, which we call nearest neighbor reducible, we presented an algorithm for constructing a compactly supported orthogonal basis, or a compactly supported Wannier basis when translationally invariant. Additionally, we showed hybrid Wannier functions that are compactly supported can be constructed for any choice of the localization axis iff. the associated projector is strictly local. Since the localization properties of band projectors and Wannier functions are closely related to band topology, we also showed that all translationally invariant strictly local projectors in systems with dimensions two and higher are topologically trivial. Additionally, using some simple arguments, we have shown that projectors without any symmetry other than possibly lattice translation symmetry, that are associated with compactly supported orthogonal bases are Chern trivial. Moreover, the existence of hybrid Wannier functions that are compactly supported along the localized axis for any choice of the localized axis implies topological triviality, unlike exponentially localized hybrid Wannier functions, which can exist even for topologically non-trivial bands.

Our results suggest a number of interesting directions for future work. The compactly supported orthogonal bases resulting from our construction may not be maximally localized. It would be interesting to improve the bounds on the spatial support of these basis functions, and also formulate an analytic procedure which results in maximally localized orthogonal basis functions. In the case of translationally invariant projectors, our procedures generate Wannier bases in supercell representations. A natural follow-up would be to find minimal sized supercell representations which have compactly supported Wannier bases. Here, we have shown that a strictly local projector and a compact orthonormal basis are essentially equivalent on 1d lattices. For dimensions two and higher, it would be useful to identify locality conditions on the projector operator that are equivalent to the existence of a compactly supported orthogonal basis. It would also be interesting to prove or disprove the topological triviality of strictly local projectors in arbitrary dimensional systems in the absence of lattice translational invariance.

\section{Acknowledgements}
We thank A. Culver, D. Reiss, X. Liu, A. Brown and L. Lindwasser for useful discussions and comments. P.S., F.H. and R.R. acknowledge support from the NSF under CAREER Grant No. DMR-1455368, and from the Mani L. Bhaumik Institute for Theoretical Physics. P.S. and R.R. acknowledge the funding support from the University of California Laboratory Fees Research Program funded by the UC Office of the President (UCOP), grant ID LFR-20-653926.

\appendix
\section{Appendices}

\subsection{Compactly Supported Wannier Functions for an Example Hamiltonian} \label{subsec:1d_example}
In this section, we discuss an example of an SL projector in 1d, and compare our results with numerical techniques from the literature. To our knowledge, there are no simple methods in the literature of proving that the span of generic SL projectors possess a compactly supported orthogonal basis (CSOB). Here, we consider the following example of a 1d two-band (translationally invariant) Hamiltonian given in k-space by
\begin{align}
H(k) = \frac{1}{6}
\begin{pmatrix}
3 \sin k+\cos k \left(2\sqrt{2} \cos k -2 \sqrt{2}\sin k +8 \sqrt{2}+3\right) & (\cos k-\sin k+4) \left(i-2 \sqrt{2} \sin k\right) \\
 -(\cos k-\sin k+4) (i+2 \sqrt{2} \sin k ) & (\cos k+2) (3-2 \sqrt{2} \cos k )+ (2 \cos k \sqrt{2}+3 ) (\sin k-2) \\
\end{pmatrix}. \label{eq:1d_example_H}
\end{align}
The two energy bands of this Hamiltonian are $E_1(k) = 2 + \cos k$, and $E_2(k) = -2  + \sin k$. The band projector corresponding to the $E_1$ band is given by
\begin{align}
\begin{split}
P(k) = \frac{1}{6} \begin{pmatrix}
 2 \cos (k) \sqrt{2}+3 & i-2 \sqrt{2} \sin (k) \\
  -i-2 \sqrt{2} \sin (k) &  3-2 \sqrt{2} \cos (k) 
\end{pmatrix}, 
\end{split}\label{eq:1d_projector_example}
\end{align}
which is an SL projector. Since the projector is translationally invariant, we seek a Wannier basis. Our theorem implies that there exist CS Wannier functions spanning each of the two bands in a size-two supercell representation.

Usually, Wannier functions are computed by first obtaining the corresponding Bloch wavefunctions. For the example projector, the Bloch wavefunction (upto a phase) is 
\begin{align}
\psi_1(k) &= \frac{1}{\sqrt{6}}\begin{pmatrix} 
\frac{2 \sqrt{2} \sin (k)-i}{\sqrt{3-2 \sqrt{2} \cos (k)}} \\
-\sqrt{3-2 \sqrt{2} \cos (k)}
\end{pmatrix}.  \label{eq:gauge_free_SL_example}
\end{align}
A corresponding Wannier basis is then obtained through a Fourier transform of the Bloch wavefunction. Specifically, a Wannier function $w_1^R(z)$ localized at cell $R$ can be obtained using 
\begin{align}
w_1^R(z) &= \frac{L}{2\pi} \int_{-\pi}^{\pi} dk e^{ik(R-z)} \psi_1(k), \label{eq:WF_is_FT_of_BW}
\end{align}
where $L$ is the system size. Since $\psi_1(k)$ is a smooth function of $k$, the corresponding Wannier function is exponentially localized (see figure \ref{fig:1d_comparison}). However, Wannier functions are not unique, because of a gauge degree of freedom:
\begin{align*}
\psi_1(k) \rightarrow e^{i\theta(k)} \psi_1(k).
\end{align*}
A more localized Wannier basis can be obtained by using a `better' gauge $\theta(k)$.
\begin{figure}[t]
    \centering
        \centering
        \includegraphics[scale=0.35]{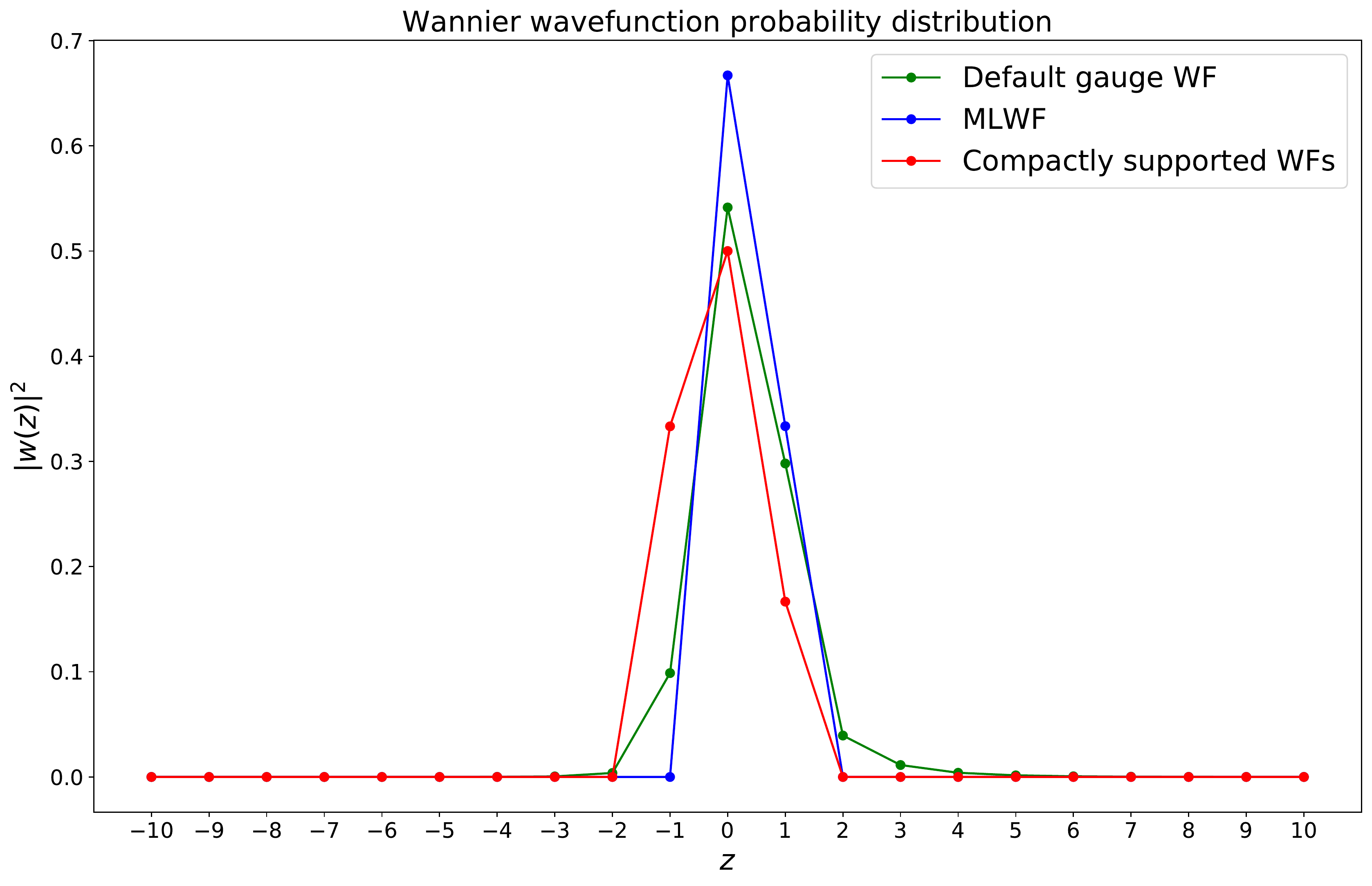}
        \caption{Wannier functions (WFs) centered at $R=0$, corresponding to three different gauge choices and corresponding to the span of \eqref{eq:1d_projector_example} are shown. The wavefunction probability is plotted a function of position. Green: Numerically obtained Wannier function corresponding to \eqref{eq:gauge_free_SL_example} is exponentially localized. Blue: Numerically obtained MLWF is compactly supported upto numerical precision. Red: The two compactly supported Wannier functions obtained analytically using algorithm \ref{algo:GS_1d_Wannier} in a size 2 supercell representation have the same probability distribution in the primitive cell representation.}
        \label{fig:1d_comparison}
\end{figure}

While it is not the objective of this paper to obtain maximally localized Wannier functions (MLWFs) for SL projectors, it is useful to compare our approach with the MLWF procedure~\cite{marzari1997maximally}. The k-space MLWF procedure converges to the optimal gauge using a gradient descent procedure on an initial `guess' Wannier basis. While our numerical experiments indicate that the MLWF procedure generates compactly supported Wannier functions (which are usually different from the ones obtained using procedure~\ref{algo:GS_1d_Wannier}), they do not guarantee the existence of compactly supported Wannier functions for arbitrary SL projectors. In contrast, our procedure is analytical, and generates wavefunctions which are exactly compactly supported.

For non-translationally invariant projectors, a notable point of similarity between the MLWF procedure and our approach is the starting point which involves choosing `trial orbitals' on which the projector is operated. In both procedures, the post-projection wavefunctions are orthonormalized (followed by gradient descent in the case of the MLWF algorithm). While the MLWF procedure uses the Lowdin (i.e. symmetric) orthogonalization procedure~\cite{lowdin1950non}, we use the Gram-Schmidt procedure. This is intentional: the Lowdin procedure is applicable only if the set to be orthogonalized is linearly independent, while the Gram-Schmidt procedure lacks this restriction.  Hence, for generic SL projectors, the real space MLWF procedure requires suitably chosen trial orbitals. While such a choice may exist, we are not aware of a general method for obtaining one for arbitrary SL projectors. A much simpler approach (which we adopt) is to choose all the orbitals $\ket{\vec{r},n}$ as trial orbitals. The set $\p_{\vec{r}}^P$ (defined in~\eqref{eq:defn piz}) which is generally linearly dependent can then be orthonormalized using the Gram-Schmidt procedure.

\subsection{A Simple Method for Constructing SL Projectors} \label{sec:NN P method}
In this section, we present a simple method for constructing translationally invariant SL projectors in one and two dimensions. We note that while one may use the equivalence that we have proven for 1d and construct SL projectors from a compactly supported translationally invariant orthogonal basis, such a method is not entirely straightforward to implement since one first needs to construct a translationally invariant orthogonal basis which is compactly supported and of an appropriate size. A much simpler alternative inspired by Clifford algebras can however be used, wherein an NN projector can be constructed by first obtaining what we call an NN \emph{flat} Hamiltonian. Such a flat Hamiltonian is an example of a flat-band Hamiltonian with all bands being flat and with energies $\pm1$. We obtain the projector $P$ onto the $-1$ eigenspace of $H$, and note that the since $H = \mathds{1} - 2P$, $P$ is guaranteed to be a nearest neighbor projector.  We start by constructing an example of a 1d NN flat-band Hamiltonian in subsection \ref{subsec:1d FB}, followed by a construction of a 2d NN flat band Hamiltonian in subsection \ref{subsec:2d FB}. While we explicitly describe the procedure only for strictly local projectors in 1 and 2 dimensional lattices, this method can be straightforwardly generalized to higher dimensions, and larger hopping distances. Although there exist SL projectors that cannot be generated using this method, it can still be used to create many interesting examples of SL projectors.

\subsubsection{1d Strictly Local Projectors} \label{subsec:1d FB}
It is convenient to utilize the Fourier space representation since we only consider translationally invariant projectors. 
First, we construct a flat NN Hamiltonian $H$, from which we will extract the desired NN projector. For $H$ to be an NN Hamiltonian, its Fourier space representation must be of the form:
\begin{align*}
H(k) = C_{+}e^{ik} + C_{0} + C_{-}e^{-ik},
\end{align*}
with $C_+$, $C_-$ and $C_0$ being $n\times n$ matrices which are constrained by the equations $H(k)^\dagger = H(k)$, and $H(k)^2 = 1$. For simplicity, we choose $H$ to possesses two bands. Consequently, $H(k)$ can be expressed in terms of the identity matrix $1_{2\times 2}$ and the two-dimensional Pauli matrices $\{\sigma_i\}$ as 
\begin{align*}
H(k) = a_0(k) 1_{2\times 2} + \sum_{i=1}^3 a_i (k) \sigma_i.
\end{align*}
Since $\{ \sigma_i, \sigma_j\}=2\delta_{ij}$, the condition that $H(k)$ is flat translates to $\sum_{\mu=0}^3 a_\mu ^2(k) = 1$ and $a_i a_0 = 0$. In order to obtain interesting solutions, we choose $a_0 = 0$. Together with $H^\dagger = H$, this implies that 
\begin{align}
\begin{split}
a_1(k)^2 + a_2(k)^2 + a_3(k)^2 &= 1; \\
a_i(k)^* &= a_i (k). 
\end{split}  \label{eq:1d_flat_H_condition}
\end{align}
Since the Hamiltonian is of an NN form, each $a_i(k)$ is expressible as
\begin{align*}
a_i(k) = c_{i} X + c_i^* X^{-1} + d_i,
\end{align*}
with complex $c_i$'s, real $d_i$'s, and $X \coloneqq e^{ik}$. Conditions \eqref{eq:1d_flat_H_condition} imply that
\begin{align}
\begin{split}
    \sum_i c_i^2 &= 0 \\
    \sum_i c_i d_i &= 0 \\
    \sum_i 2|c_i|^2 + d_i^2 &=1.
\end{split}   \label{eq:1d FB conditions c and d}
\end{align}

Solutions to these equations can be used to construct various flat Hamiltonians and projectors. A trivial example is one with $c_i=0$, and $d_1=0$, $d_2=0$ and $d_3=1$, which corresponds to
\begin{align*}
    H(k) &\equiv \begin{pmatrix}
    1 & 0 \\
    0 & -1
    \end{pmatrix}, \\
    \text{and }P(k) &= \frac{1_{2\times 2} - H(k)}{2} \equiv \begin{pmatrix}
    0 & 0 \\
    0 & 1
    \end{pmatrix}.
\end{align*}
Less trivial solutions of constraints \eqref{eq:1d FB conditions c and d} can be used to construct more interesting projectors. For example, consider the following parameters:
\begin{align*}
c_1 &= \frac{1}{3}, \ c_2 = \frac{1}{3} e^{\frac{2\pi}{3}i}, \ c_3 = \frac{1}{3} e^{\frac{4\pi}{3}i} , \  d_i = \frac{1}{3}.
\end{align*}
The corresponding Hamiltonian is given by
\begin{align}
\begin{split}
H(k) = \frac{1}{3}\begin{pmatrix}
(1+2\cos (k+\frac{4\pi}{3})) & (1+2\cos k) -i (1+2\cos (k+\frac{2\pi}{3})) \\
(1+2\cos k) +i (1+2\cos (k+\frac{2\pi}{3})) & -(1+2\cos (k+\frac{4\pi}{3})) 
\end{pmatrix}.
\end{split} \label{eq:Hk_appendix_example1d}
\end{align}
The projector $P(k)$ onto the $-1$ eigenspace can be obained by using the equation $P(k) = (1_{2\times 2} - H(k)) /2$. Since $P(k)$ has matrix elements which are Laurent polynomials in $e^{ik}$, it is an SL projector.

Having obtained an SL projector, one can use it to construct Hamiltonians which have CS Wannier functions, with any choice of the band energy (flat, or otherwise). For example one may construct a strictly local flat-band Hamiltonian, with the flat band possessing CS Wannier functions, i.e. orthogonal compact localized states (CLSs). To that end, if $P(k)$ is an SL projector obtained using the method above, we can choose it to correspond to some constant energy, say $E$. The band associated with the remaining subspace, i.e. the image of $1-P(k)$, can be chosen to have a dispersion $E(k)$, which should be chosen to be a real function expressible as a Laurent polynomial in $e^{ik}$. We can also add more bands to our Hamiltonian by constructing another Hermitian matrix $H'(k)$ with entries which are Laurent polynomials in $e^{ik}$. Arbitrary examples of $H'(k)$ and $E(k)$ satisfying the constraints mentioned above can be easily constructed. Putting it all together, we obtain a strictly local flat band Hamiltonian $\mathbb{H}(k)$ using
\begin{align}
    \mathds{H}(k) = H'(k) \oplus [E(k)(1-P(k)) + E P(k)]. \label{eq:1d FB constructed}
\end{align}

In order to construct an SL flat band Hamiltonian with a larger number of flat bands, one can use the method above to create multiple SL $P(k)$'s and assign a constant energy to each projector. Specifically, we may construct multiple flat band Hamiltonians using \eqref{eq:1d FB constructed}, and take their direct sum to construct a flat-band Hamiltonian with a larger number of flat bands. Alternatively, one may use the higher dimensional Dirac (or gamma) matrices for an analogous construction. To illustrate the latter procedure, we show how this can be used to construct nearest-neighbor projectors on 2d lattices in the next subsection.

\subsubsection{2d Strictly Local Projectors} \label{subsec:2d FB}

Similar to the previous subsection, we start with the construction of a flat Hamiltonian $H$. Here, we choose $H$ to have four bands in order to demonstrate the use of higher dimensional generators of the Clifford algebra. Hence, we express the flat Hamiltonian in terms of Dirac matrices $\Gamma^\mu$, instead of Pauli matrices, as follows:
\begin{align*}
H(\vec{k}) &= \sum_{\mu = 0}^3 a_\mu(\vec{k}) \Gamma^\mu ;\\
\text{with }\Gamma^0 = \gamma^0 &= \begin{pmatrix} \mathds{1}_2 & 0 \\
0 & -\mathds{1}_2\end{pmatrix} ,\\
\Gamma^{1} = i\gamma^1 &= i\begin{pmatrix} 0 & \sigma^x \\ -\sigma^x &0\end{pmatrix},\\
\Gamma^{2} = i\gamma^2 &= i\begin{pmatrix} 0 & \sigma^y \\ -\sigma^y &0\end{pmatrix},\\
\text{and }\Gamma^{3} = i\gamma^3 &= i\begin{pmatrix} 0 & \sigma^z \\ -\sigma^z &0\end{pmatrix}.
\end{align*}
The Dirac matrices satisfy the anti-commutation relations $\{ \Gamma^\mu,\Gamma ^\nu\} = 2\delta^{\mu \nu}$ and $\Gamma^{\mu\dagger} = \Gamma ^\mu$.
For $H(\vec{k})$ to be a nearest-neighbor Hamiltonian, the parameters $a_i(\vec{k})$ must be of the form
\begin{align}
a_i(\vec{k}) &= c_{ix} X + c_{ix}^* X^{-1}
+c_{iy} Y + c_{iy}^* Y^{-1} \nonumber\\
& \ +c_{ixy} XY + c_{ixy}^* X^{-1}Y^{-1} \nonumber\\
& \ +c_{-ixy}XY^{-1} + c_{-ixy}^* X^{-1}Y \nonumber\\
& \ +d_i,
\end{align}
with complex $c$'s, real $d$'s, and $X=e^{ik_x}$, $Y=e^{ik_y}$. $H(\vec{k})^2 = \mathds{1}$ leads to the condition:
\begin{align}
\sum_i a_i^2 = 1.
\end{align}
Equating the coefficients of all products of all powers of  $X$ and $Y$ gives us the following conditions:

\begin{align*}
\sum 2(|c_{ix}|^2 + |c_{iy}|^2 +   |c_{ixy}|^2 + |c_{-ixy}|^2 ) + d_i^2 &= 1 \\
\sum c_{iy}c_{-ixy} + c_{iy}^* c_{ixy} + c_{ix}d_i &=0 \\
\sum c_{ix} c_{-ixy}^* + c_{ix}^* c_{ixy} + c_{iy}d_i &= 0 \\
\sum c_{ix}^2 + 2c_{ixy} c_{-ixy} &= 0 \\
\sum c_{iy}^2 + 2c_{ixy} c_{-ixy}^* &= 0 \\
\sum c_{ix}c_{iy} + c_{ixy} d_i &= 0 \\
\sum c_{ix}c_{iy}^* + c_{-ixy} d_i &= 0 \\
\sum c_{ix} c_{ixy} &= 0 \\
\sum c_{ix} c_{-ixy} &= 0 \\
\sum c_{iy} c_{ixy} &= 0 \\
\sum c_{iy} c_{-ixy}^* &= 0 \\
\sum c_{ixy}^2 &= 0\\
\sum c_{-ixy}^2 &= 0.
\end{align*}

Any solution of these set of equations can used to create a projector. For example, choosing $c_{ixy} = c_{-ixy} = d_i = 0$, the following choice satisfies all the conditions:
\begin{center}
\begin{tabular}{ |c|c|c|c|c| } 
 \hline
 $\mu$ & $0$ & $1$ & $2$ & $3$ \\
 \hline
 $c_{\mu x}$ & $0$ & $0$ & $\frac{1}{2\sqrt{2}}$ & $\frac{i}{2\sqrt{2}}$  \\ 
 \hline
 $c_{\mu y}$ & $\frac{1}{2\sqrt{2}}$ & $\frac{i}{2\sqrt{2}}$ & $0$ & $0$ \\ 
 \hline
\end{tabular}
\end{center}

This corresponds to the Hamiltonian:

\begin{align*}
H(\vec{k}) = \frac{1}{2\sqrt{2}} \begin{pmatrix}
Y + Y^{-1} & 0 & -(X-X^{-1}) & -(Y-Y^{-1})+(X+X^{-1})\\
0 & Y+Y^{-1} & -(Y-Y^{-1})-(X+X^{-1}) & (X-X^{-1}) \\
X-X^{-1} & (Y-Y^{-1})-(X+X^{-1}) & -(Y+Y^{-1}) & 0 \\
(Y-Y^{-1}) & -(X-X^{-1}) &0 & -(Y+Y^{-1})
\end{pmatrix}.
\end{align*}

We obtain $P(\vec{k})$ by using $P(\vec{k}) = \frac{1_{4\times 4}-H(\vec{k})}{2}$.

\bibliographystyle{unsrt}
\bibliography{biblio}
\end{document}